\documentclass[11pt,letterpaper,reqno]{amsart}

\title{Nonholonomic Hamilton--Jacobi Theory via Chaplygin Hamiltonization}
\date{\today}
\author{Tomoki Ohsawa}
\address{Department of Mathematics, University of California, San Diego, 9500 Gilman Drive, La Jolla, CA 92093-0112}
\email{tohsawa@ucsd.edu}
\author{Oscar E. Fernandez}
\address{Institute for Mathematics and its Applications, University of Minnesota, 207 Church Street~SE, Minneapolis, MN 55455}
\email{oscarum@umich.edu, ferna007@ima.umn.edu}
\author{Anthony M. Bloch}
\address{Department of Mathematics, University of Michigan, 530 Church Street, Ann Arbor, MI 48109-1043}
\email{abloch@umich.edu}
\author{Dmitry V. Zenkov}
\address{Department of Mathematics, North Carolina State University, Raleigh, NC 27695}
\email{dvzenkov@ncsu.edu}

\usepackage{amsmath,amsthm,amssymb,graphicx,paralist}
\usepackage[margin=1in]{geometry}

\usepackage[all]{xy}

\usepackage[numbers,sort&compress]{natbib}

\usepackage[colorlinks=false]{hyperref}

\numberwithin{equation}{section}

\theoremstyle{plain}
\newtheorem{theorem}{Theorem}[section]
\newtheorem{corollary}[theorem]{Corollary}
\newtheorem{lemma}[theorem]{Lemma}
\newtheorem{proposition}[theorem]{Proposition}

\theoremstyle{definition}
\newtheorem{definition}[theorem]{Definition}
\newtheorem{example}[theorem]{Example}

\theoremstyle{remark}
\newtheorem{remark}[theorem]{Remark}

\def\od#1#2{\dfrac{d#1}{d#2}}
\def\pd#1#2{\dfrac{\partial #1}{\partial #2}}
\def\tpd#1#2{\partial #1/\partial #2}
\def\tod#1#2{d#1/d#2}

\def\parentheses#1{\!\left(#1\right)}
\def\brackets#1{\!\left[#1\right]}
\def\braces#1{\!\left\{#1\right\}}

\def\R{\mathbb{R}}
\def\defeq{\mathrel{\mathop:}=}
\def\eqdef{=\mathrel{\mathop:}}
\def\setdef#1#2{ \left\{ #1 \ |\ #2 \right\} }
\def\ip#1#2{\left\langle#1,#2\right\rangle}
\def\divergence{\mathop{\mathrm{div}}\nolimits}
\def\hor{\mathop{\mathrm{hor}}\nolimits}
\def\ver{\mathop{\mathrm{ver}}\nolimits}
\def\hl{\mathop{\mathrm{hl}}\nolimits}
\def\id{\mathop{\mathrm{id}}\nolimits}

\def\FL{\mathbb{F}L}
\def\F{\mathbb{F}}
\def\Ad{\mathop{\mathrm{Ad}}\nolimits}

\begin{document}

\footskip=.6in

\begin{abstract}
  We develop Hamilton--Jacobi theory for Chaplygin systems, a certain class of nonholonomic mechanical systems with symmetries, using a technique called Hamiltonization, which transforms nonholonomic systems into Hamiltonian systems.
  We give a geometric account of the Hamiltonization, identify necessary and sufficient conditions for Hamiltonization, and apply the conventional Hamilton--Jacobi theory to the Hamiltonized systems.
  We show, under a certain sufficient condition for Hamiltonization, that the solutions to the Hamilton--Jacobi equation associated with the Hamiltonized system also solve the nonholonomic Hamilton--Jacobi equation associated with the original Chaplygin system.
  The results are illustrated through several examples.
\end{abstract}

\maketitle


\section{Introduction}
\subsection{Background and Motivation}
In 1911 S.A. Chaplygin published a paper (re-published in English in \cite{Ch2008})  introducing his theory of the ``reducing multiplier'' into the study of nonholonomically constrained mechanical systems.
In his paper, Chaplygin showed that a two degree of freedom nonholonomic system possessing an invariant measure became Hamiltonian after a suitable reparameterization of time, a process we would like to refer to as {\em Chaplygin Hamiltonization}.
Since then, Chaplygin's result has generated considerable interest and been extended \cite{St1989,KoRiEh2002,FeJo2004,HoGa2009,FeMeBl2009} to more general settings.

However, a second contribution contained in Chaplygin's paper has been left undeveloped.
In Section~5 of his paper, Chaplygin integrates the nonholonomic system now known as the Chaplygin Sleigh~\cite{Bl2003} by using the Hamilton--Jacobi equation for the Hamiltonized system.
The aim of this paper is to develop this idea further to establish a link with the nonholonomic Hamilton--Jacobi equation in \citet{IgLeMa2008} and \citet{OhBl2009}.

Specifically, we first employ the technique called Chaplygin Hamiltonization to transform Chaplygin systems into Hamiltonian systems, and then apply the conventional Hamilton--Jacobi theory to the resulting Hamiltonian systems to obtain what we would like to call the {\em Chaplygin Hamilton--Jacobi equation}.
This is an indirect approach towards Hamilton--Jacobi theory for nonholonomic systems, compared to the direct approach of extending Hamilton--Jacobi theory to nonholonomic systems, as in \citet{IgLeMa2008}, \citet{LeMaMa2010}, \citet{OhBl2009}, and \citet{CaGrMaMaMuRo2010}.

\subsection{Direct vs.~Indirect Approaches}
The indirect approach to nonholonomic Hamilton--Jacobi theory via Chaplygin Hamiltonization has both advantages and disadvantages.
The main advantage is that we have a conventional Hamilton--Jacobi equation and thus the separation of variables argument applies in a rather straightforward manner compared to the direct approach in \citet{OhBl2009}.
A disadvantage is that the Chaplygin Hamiltonization works only for certain nonholonomic systems; even if it does, the relationship between the Hamilton--Jacobi equation and the original nonholonomic system is not transparent, since one has to inverse-transform the information in the Hamiltonized systems.
Nevertheless, Hamiltonization is known to be a powerful technique for integration of nonholonomic systems~\cite{Ch2008,FeJo2004,FeJo2009,BlFeMe2009,FeMeBl2009}, and hence it is interesting to establish a connection with the direct approach.

Let us briefly summarize the differences between two approaches.
Recall from \citet{OhBl2009} that the {\em nonholonomic Hamilton--Jacobi equation} is an equation for a one-form $\gamma$ on the original configuration manifold $Q$:
\begin{equation}
  \label{eq:NHHJ}
  H \circ \gamma = E,
\end{equation}
along with the condition that $\gamma$, seen as a map from $Q$ to $T^{*}Q$, takes values in the constrained momentum space $\mathcal{M} \subset T^{*}Q$ (see Eq.~\eqref{eq:mathcalM} below), i.e., $\gamma: Q \to \mathcal{M}$, and also that
\begin{equation}
  \label{eq:dgamma}
  d\gamma|_{\mathcal{D}\times\mathcal{D}} = 0, \text{ i.e., } d\gamma(v,w) = 0 \text{ for any } v, w \in \mathcal{D},
\end{equation}
where $\mathcal{D} \subset TQ$ is the distribution defined by nonholonomic constraints, and $H: T^{*}Q \to \R$ the Hamiltonian.

On the other hand, the Chaplygin Hamiltonization first reduces the system by identifying it as a so-called Chaplygin system with a symmetry group $G$, and then Hamiltonizes the system on the cotangent bundle $T^{*}(Q/G)$ of the reduced configuration space $Q/G$.
The resulting system is a (strictly) Hamiltonian system on $T^{*}(Q/G)$ with another Hamiltonian $\bar{H}_{\rm C}: T^{*}(Q/G) \to \R$; so we may apply the conventional Hamilton--Jacobi theory to the Hamiltonized system to obtain the {\em Chaplygin Hamilton--Jacobi equation}
\begin{equation*}
  \bar{H}_{\rm C} \circ d\bar{W} = E,
\end{equation*}
which is a partial differential equation for a function $\bar{W}: Q/G \to \R$.
Therefore, the difference lies not only in the forms of the equations (the former involves the one-form $\gamma$, which is not even closed, whereas the latter invokes the exact one-form $d\bar{W}$), but also in the spaces on which the equations are defined.
Furthermore, the Chaplygin Hamilton--Jacobi equation corresponds to the Hamiltonized dynamics and is related to the original nonholonomic one in a rather indirect way.
Therefore, on the surface, there does not seem to be an apparent relationship between the two approaches.

\subsection{Main Results}
The main goal of this paper is to establish a link between the two distinct approaches towards Hamilton--Jacobi theory for nonholonomic systems.
To that end, we first formulate the Chaplygin Hamiltonization in an intrinsic manner to elucidate the geometry involved in the Hamiltonization.
This gives a slight generalization of the Chaplygin Hamiltonization by \citet{FeJo2004} and also an intrinsic account of the necessary and sufficient condition for Hamiltonizing a Chaplygin system presented in \cite{FeMeBl2009}.
These results are also related to the existence of an invariant measure in nonholonomic systems~(see, e.g., \citet{Ko2002}, \citet{ZeBl2003}, and \citet{FeJo2004}).

We also identify a sufficient condition for the Chaplygin Hamiltonization, which turns out to be identical to one of those for another kind of Hamiltonization (which renders the systems ``conformal symplectic''~\cite{HoGa2009}) obtained by \citet{St1989} and \citet{CaCoLeMa2002}.
We then give an explicit formula that transforms the solutions of the Chaplygin Hamilton--Jacobi equation into those of the nonholonomic Hamilton--Jacobi equation~(see Fig.~\ref{fig:RelatingHJs}).
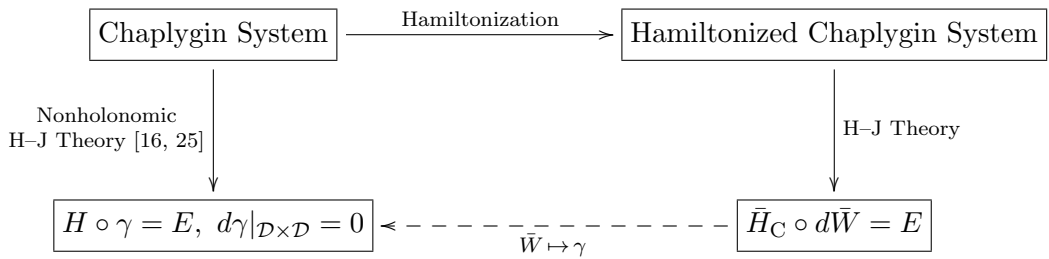
\begin{figure}[h!]
  \centerline{
    \xymatrix@!0@R=1in@C=3.25in{
      {\framebox{\strut Chaplygin System}} \ar[r]^{\text{Hamiltonization}\qquad\quad} \ar[d]_{\substack{\text{Nonholonomic}\smallskip\\\text{H--J Theory~\cite{IgLeMa2008,OhBl2009}}}} & \framebox{\strut Hamiltonized Chaplygin System} \ar[d]^{\text{H--J Theory}}
      \\
      \framebox{\strut $H \circ \gamma = E,\ d\gamma|_{\mathcal{D} \times \mathcal{D}} = 0$} \ar@{<--}[r]_{\qquad \bar{W}\,\mapsto\,\gamma} & \framebox{\strut $\bar{H}_{\text{C}} \circ d\bar{W} = E$}
    }
  }
  \caption{Relationship between the nonholonomic H--J equation applied to a Chaplygin system and the H--J equation applied to the Hamiltonized Chaplygin system. Explicit formulas for the correspondence $\bar{W} \mapsto \gamma$ are given in Theorems~\ref{thm:ChaplyginHJ-NHHJ} and \ref{thm:ChaplyginHJ2-NHHJ}.}
  \label{fig:RelatingHJs}
\end{figure}
Interestingly, it turns out that the sufficient condition plays an important role here as well.
We also present an extension of these results to a class of systems that are Hamiltonizable after reduction by two stages, following the idea of \citet{HoGa2009}.
We illustrate, through several examples, that the Chaplygin Hamilton--Jacobi equation may be solved by separation of variables, and that the solutions are identical to those obtained by \citet{OhBl2009} after the transformation mentioned above.

\subsection{Outline}
We begin with an overview of nonholonomic mechanical systems in Section~\ref{hfnm}, specializing to Chaplygin systems in Section~\ref{css}.
After discussing the relationship between the Hamiltonizability of a nonholonomic system and the existence of an invariant measure for it in Section~\ref{chsc}, we derive necessary and sufficient conditions for a Chaplygin system to be Hamiltonizable in an intrinsic manner in Section~\ref{chsc2}.
This result then leads to the development of Hamilton--Jacobi theory for Hamiltonizable Chaplygin systems in Section~\ref{hjnhsec}.
Specifically, we relate the Chaplygin Hamilton--Jacobi equation for the Hamiltonized system with the nonholonomic Hamilton--Jacobi equation for the original system.
A couple of examples are presented in Section~\ref{exsec} to illustrate the theoretical results.
In Section~\ref{sec:FurtherReduction}, we introduce a further reduction of the reduced Chaplygin systems under certain conditions; the second reduction is employed to Hamiltonize those systems that are not Hamiltonizable after the first reduction.
Then, in Section~\ref{sec:NHHJafter2ndReduction}, we relate the Chaplygin Hamilton--Jacobi equation for such systems with the nonholonomic Hamilton--Jacobi equation.
We then illustrate the theory in the Snakeboard example.


\section{Chaplygin Systems}
\subsection{Hamiltonian Formulation of Nonholonomic Mechanics}\label{hfnm}
Consider a nonholonomic system on an $n$-dimensional configuration manifold $Q$ with a constraint distribution $\mathcal{D} \subset TQ$ defined by the constraint one-forms $\{ \omega^{s} \}_{s=1}^{m}$ as
\begin{equation*}
  \mathcal{D} = \setdef{ v \in TQ }{ \omega^{s}(v) = 0,\, s = 1, \dots, m }
\end{equation*}
and also with the Lagrangian $L: TQ \to \R$ of the form
\begin{equation}
  \label{eq:SimpleLagrangian}
  L(v_{q}) = \frac{1}{2}g_{q}(v_{q}, v_{q}) - V(q)
\end{equation}
with the kinetic energy metric $g$ defined on $Q$.
Define the Legendre transform $\FL: TQ \to T^{*}Q$ by
\begin{equation*}
  \ip{ \FL(v_{q}) }{ w_{q} } = \left.\od{}{\varepsilon} L(v_{q} + \varepsilon\,w_{q}) \right|_{\varepsilon=0}
  = g_{q}(v_{q}, w_{q})
  = \ip{ g^{\flat}_{q}(v_{q}) }{w_{q}},
\end{equation*}
where the last equality defines $g^{\flat}: TQ \to T^{*}Q$; hence we have $\FL = g^{\flat}$.
Also define the Hamiltonian $H: T^{*}Q \to \R$ by
\begin{equation*}
  H(p_{q}) \defeq \ip{p_{q}}{v_{q}} - L(v_{q}),
\end{equation*}
where $v_{q} = (\FL)^{-1}(p_{q})$ on the right-hand side.
Then, Hamilton's equations for nonholonomic systems are written as follows:
\begin{equation}
  \label{eq:NHHam}
  i_{X} \Omega = dH - \lambda_{s} \pi_{Q}^{*} \omega^{s},
\end{equation}
along with
\begin{equation}
  \label{eq:NHHam-condition}
  T\pi_{Q}(X) \in \mathcal{D}
  \quad\text{or}\quad
  \omega^{s}(T\pi_{Q}(X)) = 0
  \quad \text{for}\quad s = 1, \dots, m,
\end{equation}
where $\pi_{Q}: T^{*}Q \to Q$ is the cotangent bundle projection.
Introducing the {\em constrained momentum space}
\begin{equation}
  \label{eq:mathcalM}
  \mathcal{M} \defeq \FL(\mathcal{D}) \subset T^{*}Q,
\end{equation}
the above constraints may be replaced by $p \in \mathcal{M}$.

\subsection{Chaplygin Systems}\label{css}
\begin{definition}[Chaplygin Systems]
  \label{def:ChaplyginSystems}
  A nonholonomic system with Hamiltonian $H$ and distribution $\mathcal{D}$ is called a {\em Chaplygin system} if there exists a Lie group $G$ and a free and proper group action of it on $Q$, i.e., $\Phi: G \times Q \to Q$ or $\Phi_{h}: Q \to Q$ for any $h \in G$, such that
  \begin{enumerate}[(i)]
  \item the Hamiltonian $H$ and the distribution $\mathcal{D}$ are invariant under the $G$-action;
  \item for each $q \in Q$, the tangent space $T_{q}Q$ is the direct sum of the constraint distribution and the tangent space to the orbit of the group action, i.e.,
    \begin{equation*}
      T_{q}Q = \mathcal{D}_{q} \oplus T_{q}\mathcal{O}_{q},
    \end{equation*}
    where $\mathcal{O}_{q}$ is the orbit through $q$ of the $G$-action on $Q$, i.e.,
    \begin{equation*}
      \mathcal{O}_{q} \defeq \setdef{ \Phi_{h}(q) \in Q }{ h \in G }.
    \end{equation*}
  \end{enumerate}
\end{definition}
This setup gives rise to the principal bundle
\begin{equation*}
  \pi: Q \to Q/G \eqdef \bar{Q}
\end{equation*}
and the connection 
\begin{equation*}
  \mathcal{A}: TQ \to \mathfrak{g},
\end{equation*}
with $\mathfrak{g}$ being the Lie algebra of $G$ such that $\ker\mathcal{A} = \mathcal{D}$.
So the above decomposition may be written as
\begin{equation*}
  T_{q}Q = \ker\mathcal{A}_{q} \oplus \ker T_{q}\pi.
\end{equation*}
Furthermore, for any $q \in Q$ and $\bar{q} \defeq \pi(q) \in \bar{Q}$, the map $T_{q}\pi|_{\mathcal{D}_{q}}: \mathcal{D}_{q} \to T_{\bar{q}}\bar{Q}$ is a linear isomorphism, and hence we have the horizontal lift
\begin{equation*}
  \hl^{\mathcal{D}}_{q}: T_{\bar{q}}\bar{Q} \to \mathcal{D}_{q};
  \quad
  v_{\bar{q}} \mapsto (T_{q}\pi|_{\mathcal{D}_{q}})^{-1}(v_{\bar{q}}).
\end{equation*}
We will occasionally use the following shorthand notation for horizontal lifts:
\begin{equation*}
  v^{\rm h}_{q} \defeq \hl^{\mathcal{D}}_{q}(v_{\bar{q}}).
\end{equation*}
Therefore, any vector $W_{q} \in T_{q}Q$ can be decomposed into the horizontal and vertical parts as follows:
\begin{subequations}
  \label{eq:hor-ver_decomposition}
  \begin{equation}
    W_{q} = \hor(W_{q}) + \ver(W_{q}),
  \end{equation}
  with
  \begin{equation}
    \hor(W_{q}) = \hl^{\mathcal{D}}_{q}(w_{\bar{q}}),
    \qquad
    \ver(W_{q}) = (\mathcal{A}_{q}(W_{q}))_{Q}(q),
  \end{equation}
\end{subequations}
where $w_{\bar{q}} \defeq T_{q}\pi(W_{q})$ and $\xi_{Q} \in \mathfrak{X}(Q)$ is the infinitesimal generator of $\xi \in \mathfrak{g}$.

We may then define the reduced Lagrangian
\begin{subequations}
  \label{eq:ReducedLagrangian}
  \begin{equation}
    \bar{L} \defeq L \circ \hl^{\mathcal{D}},
  \end{equation}
  or more explicitly,
  \begin{equation}
    \bar{L}: T\bar{Q} \to \R;
    \quad
    v_{\bar{q}} \mapsto \frac{1}{2} \bar{g}_{\bar{q}}(v_{\bar{q}}, v_{\bar{q}}) - \bar{V}(\bar{q}),
  \end{equation}
\end{subequations}
where $\bar{g}$ is the metric on the reduced space $\bar{Q}$ induced by $g$ as follows:
\begin{equation}
  \label{eq:barg}
  \bar{g}_{\bar{q}}(v_{\bar{q}}, w_{\bar{q}})
  \defeq g_{q}\parentheses{ \hl^{\mathcal{D}}_{q}(v_{\bar{q}}), \hl^{\mathcal{D}}_{q}(w_{\bar{q}}) }
  = g_{q}( v^{\rm h}_{q}, w^{\rm h}_{q}),
\end{equation}
and the reduced potential $\bar{V}: \bar{Q} \to \R$ is defined such that $V = \bar{V} \circ \pi$.

This geometric structure is carried over to the Hamiltonian side~(see \citet{EhKoMoRi2004}).
Specifically, we define the horizontal lift $\hl^{\mathcal{M}}_{q}: T_{\bar{q}}^{*}\bar{Q} \to \mathcal{M}_{q}$ by
\begin{equation}
  \label{eq:hl^M}
  \hl^{\mathcal{M}}_{q}
  \defeq \FL_{q} \circ \hl^{\mathcal{D}}_{q} \circ (\F\bar{L})^{-1}_{\bar{q}}
  = g_{q}^{\flat} \circ \hl^{\mathcal{D}}_{q} \circ (\bar{g}^{\flat})^{-1}_{\bar{q}},
\end{equation}
or the diagram below commutes.
\begin{equation*}
  \vcenter{
    \xymatrix@!0@R=0.75in@C=1.1in{
      \mathcal{D}_{q} \ar[r]^{\FL_{q}} & \mathcal{M}_{q}
      \\
      T_{\bar{q}}\bar{Q} \ar[u]^{\hl^{\mathcal{D}}_{q}} & T_{\bar{q}}^{*}\bar{Q} \ar[l]^{\ (\mathbb{F}\bar{L})^{-1}_{\bar{q}}} \ar@{-->}[u]_{\hl^{\mathcal{M}}_{q}}
    }
  }
\end{equation*}
We will use the shorthand notation
\begin{equation*}
  \alpha^{\rm h}_{q} \defeq \hl^{\mathcal{M}}_{q}(\alpha_{\bar{q}})
\end{equation*}
for any $\alpha_{\bar{q}} \in T^{*}_{\bar{q}}\bar{Q}$.

We also define the reduced Hamiltonian $\bar{H}: T^{*}\bar{Q} \to \R$ by
\begin{equation}
  \label{eq:barH}
  \bar{H} \defeq H \circ \hl^{\mathcal{M}}.
\end{equation}
It is easy to check that this definition coincides with the following one by using the reduced Lagrangian $\bar{L}$:
\begin{equation*}
  \bar{H}(p_{\bar{q}}) \defeq \ip{ p_{\bar{q}} }{ v_{\bar{q}} } - \bar{L}(v_{\bar{q}}),
\end{equation*}
with $v_{\bar{q}} = (\F\bar{L})^{-1}_{\bar{q}}(p_{\bar{q}})$.

Performing the nonholonomic reduction of \citet{Ko1992} (see also \citet{BaSn1993}, \citet{EhKoMoRi2004}, and \citet{HoGa2009}), we obtain the reduced Hamilton's equations for Chaplygin systems defined by
\begin{equation}
  \label{eq:ReducedChaplyginSystem}
  i_{\bar{X}} \bar{\Omega}^{\rm nh} = d\bar{H}
\end{equation}
with the almost symplectic form
\begin{equation}
  \label{eq:barOmega_nh}
  \bar{\Omega}^{\rm nh} \defeq \bar{\Omega} - \Xi,
\end{equation}
where $\bar{X}$ is a vector field on $T^{*}\bar{Q}$ and $\bar{\Omega}$ is the standard symplectic form on $T^{*}\bar{Q}$; the two-form $\Xi$ on $T^{*}\bar{Q}$ is defined as follows:
For any $\alpha_{\bar{q}} \in T_{\bar{q}}^{*}\bar{Q}$ and $\mathcal{Y}_{\alpha_{\bar{q}}}, \mathcal{Z}_{\alpha_{\bar{q}}} \in T_{\alpha_{\bar{q}}}T^{*}\bar{Q}$, let $Y_{\bar{q}} \defeq T\pi_{\bar{Q}}(\mathcal{Y}_{\alpha_{\bar{q}}})$ and $Z_{\bar{q}} \defeq T\pi_{\bar{Q}}(\mathcal{Z}_{\alpha_{\bar{q}}})$ where $\pi_{\bar{Q}}: T^{*}\bar{Q} \to \bar{Q}$ is the cotangent bundle projection, and then set
\begin{align}
  \Xi_{\alpha_{\bar{q}}}(\mathcal{Y}_{\alpha_{\bar{q}}}, \mathcal{Z}_{\alpha_{\bar{q}}})
  &\defeq \ip{{\bf J} \circ \hl^{\mathcal{M}}_{q} (\alpha_{\bar{q}})}{ \mathcal{B}_{q}\parentheses{ \hl^{\mathcal{D}}_{q}(Y_{\bar{q}}), \hl^{\mathcal{D}}_{q}(Z_{\bar{q}}) } }
  \nonumber\\
  &= \ip{{\bf J}(\alpha^{\rm h}_{q})}{\mathcal{B}_{q}(Y^{\rm h}_{q}, Z^{\rm h}_{q})},
  \label{eq:Xi-def}
\end{align}
where ${\bf J}: T^{*}Q \to \mathfrak{g}^{*}$ is the momentum map corresponding to the $G$-action, and $\mathcal{B}$ is the curvature two-form of the connection $\mathcal{A}$.
This is well-defined, since the $\Ad^{*}$-equivariance of the momentum map ${\bf J}$ and the $\Ad$-equivariance of the curvature $\mathcal{B}$ cancel each other~\citep{KoRiEh2002}: Writing $h q \defeq \Phi_{h}(q)$, we have, using Lemma~\ref{lem:G-hl^M} and the $G$-equivariance of the momentum map ${\bf J}$ and the curvature $\mathcal{B}$~(see, e.g., \citet[Corollary~2.1.11]{MaMiOrPeRa2007} for the latter),
\begin{align*}
  \ip{{\bf J}(\alpha^{\rm h}_{h q})}{\mathcal{B}_{h q}(Y^{\rm h}_{h q}, Z^{\rm h}_{h q})}
  &= \ip{ {\bf J}\parentheses{ T^{*}_{q}\Phi_{h^{-1}}(\alpha^{\rm h}_{q}) } }{\Phi_{h}^{*}\mathcal{B}_{q}(Y^{\rm h}_{q}, Z^{\rm h}_{q})}
  \\
  &= \ip{ \Ad^{*}_{h^{-1}} {\bf J}(\alpha^{\rm h}_{q}) }{\Ad_{h} \mathcal{B}_{q}(Y^{\rm h}_{q}, Z^{\rm h}_{q})}
  \\
  &= \ip{ {\bf J}(\alpha^{\rm h}_{q}) }{\mathcal{B}_{q}(Y^{\rm h}_{q}, Z^{\rm h}_{q})}.
\end{align*}


\section{Chaplygin Hamiltonization of Nonholonomic Systems}\label{chsb}
This section discusses the so-called Chaplygin Hamiltonization of the reduced dynamics defined by Eq.~\eqref{eq:ReducedChaplyginSystem}.
The results here are mostly a summary of some of the key results of \citet{St1989}, \citet{CaCoLeMa2002}, \citet{FeJo2004}, and \citet{FeMeBl2009}.
However, our exposition is slightly different from them, and also touches on those aspects that are not found in the above papers.
Furthermore, our intrinsic account of the Hamiltonization provides us with a better understanding of the geometry involved in it, and then leads us to our main results on nonholonomic Hamilton--Jacobi theory in Sections~\ref{sec:NHHJ} and \ref{sec:NHHJafter2ndReduction}.

\subsection{Hamiltonization and Existence of Invariant Measure}\label{chsc}
We first discuss the relationship between Hamiltonization and existence of an invariant measure for nonholonomic systems.
The next subsection will show how to Hamiltonize the reduced system, Eq.~\eqref{eq:ReducedChaplyginSystem}, explicitly.

Let $f: T^{*}\bar{Q} \to \R$ be a smooth nowhere-vanishing function that is constant on each fiber, i.e., $f(\alpha_{\bar{q}}) = f(\beta_{\bar{q}})$ for any $\alpha_{\bar{q}}, \beta_{\bar{q}} \in T_{\bar{q}}^{*}\bar{Q}$.
Therefore, we can write, with a slight abuse of notation, $f(\alpha_{\bar{q}}) = f(\bar{q})$; so $f$ may be seen as a function on $\bar{Q}$ as well.
\begin{remark}
  The above definition of the function $f$ is essentially the same as that of \citet{Ch2008}, where $f$ is defined as a function on $Q$.
  However, in the present work, it is more convenient to formally define $f$ as a function on $T^{*}Q$.
\end{remark}
\begin{remark}
  In the discussion to follow, we derive certain conditions on the function $f$ in order to Hamiltonize the system given by Eq.~\eqref{eq:ReducedChaplyginSystem}.
  It sometimes turns out that such $f$ is nowhere-vanishing only on an open subset $U$ in $\bar{Q}$.
  In such cases, we redefine $\bar{Q} \defeq U$.
\end{remark}
Now, consider the vector field
\begin{equation*}
  \bar{X}/f = \frac{1}{f}\bar{X} \in \mathfrak{X}(T^{*}\bar{Q}),
\end{equation*}
and let $\Phi_{t}^{\bar{X}/f}: T^{*}\bar{Q} \to T^{*}\bar{Q}$ be the flow defined by the corresponding vector field, i.e., for any $\alpha_{\bar{q}} \in T^{*}\bar{Q}$, 
\begin{equation*}
  \left.\od{}{t}\Phi_{t}^{\bar{X}/f}(\alpha_{\bar{q}})\right|_{t=0} = (\bar{X}/f)(\alpha_{\bar{q}}) = \frac{1}{f(\alpha_{\bar{q}})}\,\bar{X}(\alpha_{\bar{q}}).
\end{equation*}
Furthermore, define a map $\Psi_{\!f}: T^{*}\bar{Q} \to T^{*}\bar{Q}$ by
\begin{equation*}
  \Psi_{\!f}: \alpha \mapsto f \alpha,
\end{equation*}
which is clearly a diffeomorphism with the inverse $\Psi_{\!f}^{-1} = \Psi_{1/f}: T^{*}\bar{Q} \to T^{*}\bar{Q};\ \alpha \mapsto \alpha/f$, and define $\Phi_{t}^{\bar{X}_{\rm C}}: T^{*}\bar{Q} \to T^{*}\bar{Q}$ by
\begin{equation*}
  \Phi_{t}^{\bar{X}_{\rm C}}
  \defeq \Psi_{\!f} \circ \Phi_{t}^{\bar{X}/f} \circ \Psi_{\!f}^{-1}
  = \Psi_{\!f} \circ \Phi_{t}^{\bar{X}/f} \circ \Psi_{1/f},
\end{equation*}
or the diagram below commutes.
\begin{equation}
  \vcenter{
    \xymatrix@!0@R=0.75in@C=1in{
      T^{*}\bar{Q} \ar[r]^{\Phi_{t}^{\bar{X}/f}} & T^{*}\bar{Q} \ar[d]^{\Psi_{\!f}}
      \\
      T^{*}\bar{Q} \ar[u]^{\Psi_{\!f}^{-1} = \Psi_{1/f}} \ar@{-->}[r]_{\Phi_{t}^{\bar{X}_{\rm C}}} & T^{*}\bar{Q}
    }
  }
  \qquad
  \vcenter{
    \xymatrix@!0@R=0.75in@C=1in{
      \alpha/f \ar@{|->}[r] & \Phi_{t}^{\bar{X}/f}(\alpha/f) \ar@{|->}[d]
      \\
      \alpha \ar@{|->}[u] \ar@{|-->}[r] & \Phi_{t}^{\bar{X}_{\rm C}}(\alpha)
    }
  }
\end{equation}
Then, we have the vector field $\bar{X}_{\rm C} \in \mathfrak{X}(T^{*}\bar{Q})$ corresponding to the flow $\Phi_{t}^{\bar{X}_{\rm C}}$, which is the pull-back of $\bar{X}/f$ by $\Psi_{\!f}^{-1} = \Psi_{1/f}$: For any $\alpha_{\bar{q}} \in T^{*}\bar{Q}$,
\begin{align}
  \label{eq:barX_C}
  \bar{X}_{\rm C}(\alpha_{\bar{q}})
  &\defeq \left. \od{}{t} \Phi_{t}^{\bar{X}_{\rm C}}(\alpha_{\bar{q}}) \right|_{t=0}
  \nonumber\\
  &= \left. \od{}{t} \Psi_{\!f} \circ \Phi_{t}^{\bar{X}/f} \circ \Psi_{\!f}^{-1}(\alpha_{\bar{q}}) \right|_{t=0}
  \nonumber\\
  &= T\Psi_{\!f} \cdot (\bar{X}/f)(\Psi_{\!f}^{-1}(\alpha_{\bar{q}}))
  \nonumber\\
  &= (\Psi_{\!f}^{-1})^{*} (\bar{X}/f) (\alpha_{\bar{q}})
  \nonumber\\
  &= \Psi_{1/f}^{*} (\bar{X}/f) (\alpha_{\bar{q}}).
\end{align}
In particular, the third line in the above equation shows that $\bar{X}/f$ and $\bar{X}_{\rm C}$ are $\Psi_{\!f}$-related:
\begin{equation}
  \label{eq:Psi_f-related}
  T\Psi_{\!f} \circ (\bar{X}/f) = \bar{X}_{\rm C} \circ \Psi_{\!f}.
\end{equation}

Now, we relate relate the (possible) symplecticity of the vector field $\bar{X}_{\rm C}$ with the existence of an invariant measure for the reduced system, Eq.~\eqref{eq:ReducedChaplyginSystem}:
\begin{theorem}
  \label{thm:Generalized_FeJo2004}
  If $\bar{X}_{\rm C} \in \mathfrak{X}(T^{*}\bar{Q})$ is symplectic, i.e., $\pounds_{\bar{X}_{\rm C}} \bar{\Omega} = 0$, then the reduced system, Eq.~\eqref{eq:ReducedChaplyginSystem}, has the invariant measure $f^{\bar{n}-1} \bar{\Lambda}$, where $\bar{n} \defeq \dim\bar{Q}$ and $\bar{\Lambda}$ is the Liouville volume form
  \begin{equation*}
    \bar{\Lambda} \defeq \frac{(-1)^{\bar{n}(\bar{n}-1)/2}}{\bar{n}!} \underbrace{\bar{\Omega} \wedge \dots \wedge \bar{\Omega}}_{\bar{n}}
    = dq^{1} \wedge \dots \wedge dq^{\bar{n}} \wedge dp_{1} \wedge \dots \wedge dp_{\bar{n}}.
  \end{equation*}
  In other words, we have
  \begin{equation*}
    \pounds_{\bar{X}} (f^{\bar{n}-1} \bar{\Lambda}) = 0.
  \end{equation*}
\end{theorem}

This theorem is a slight generalization of the following:
\begin{corollary}[\citet{FeJo2004}]
\label{cor:FeJo2004}
  If $\bar{X}_{\rm C} \in \mathfrak{X}(T^{*}\bar{Q})$ is Hamiltonian, i.e.,
  \begin{equation*}
    i_{\bar{X}_{\rm C}} \bar{\Omega} = d\bar{H}_{\rm C}    
  \end{equation*}
  for some $\bar{H}_{\rm C}: T^{*}\bar{Q} \to \R$, then the reduced nonholonomic dynamics, Eq.~\eqref{eq:ReducedChaplyginSystem}, has the invariant measure $f^{\bar{n}-1} \bar{\Lambda}$.
\end{corollary}

\begin{proof}
  Follows easily from Cartan's formula:
  \begin{equation*}
    \pounds_{\bar{X}_{\rm C}} \bar{\Omega} = d(i_{\bar{X}_{\rm C}} \bar{\Omega}) + i_{\bar{X}_{\rm C}} d\bar{\Omega} = dd\bar{H}_{\rm C} = 0. \qedhere
  \end{equation*}
\end{proof}

We state a couple of lemmas before proving Theorem~\ref{thm:Generalized_FeJo2004}.
\begin{lemma}
  \label{lem:f_Omega}
  Let $f: T^{*}\bar{Q} \to \R$ be a smooth function that is constant on each fiber, i.e., $f(\alpha_{\bar{q}}) = f(\beta_{\bar{q}})$ for any $\alpha_{\bar{q}}, \beta_{\bar{q}} \in T_{\bar{q}}^{*}\bar{Q}$.
  Then,
  \begin{equation*}
    \underbrace{(\Psi_{\!f}^{*}\bar{\Omega}) \wedge \dots \wedge (\Psi_{\!f}^{*}\bar{\Omega})}_{\bar{n}} = f^{\bar{n}}\underbrace{\bar{\Omega} \wedge \dots \wedge \bar{\Omega}}_{\bar{n}}.
  \end{equation*}
\end{lemma}

\begin{proof}
  Let $\bar{\Theta}$ be the symplectic one-form on $T^{*}\bar{Q}$, i.e., $\bar{\Omega} = -d\bar{\Theta}$.
  Let us first calculate $\Psi_{\!f}^{*}\bar{\Theta}$: We have, for any $\alpha \in T^{*}\bar{Q}$ and $v \in T_{\alpha}T^{*}\bar{Q}$,
  \begin{align*}
    (\Psi_{\!f}^{*}\bar{\Theta})_{\alpha}(v)
    &= \bar{\Theta}_{\Psi_{\!f}(\alpha)}( T\Psi_{\!f}(v) )
    \\
    &= \ip{ \Psi_{\!f}(\alpha) }{ T\pi_{Q} \circ T\Psi_{\!f}(v) }
    \\
    &=  \ip{ f \alpha }{ T(\pi_{Q} \circ \Psi_{\!f})(v) }
    \\
    &=  f \ip{ \alpha }{ T\pi_{Q}(v) }
    \\
    &= f\, \bar{\Theta}_{\alpha}(v),
  \end{align*}
  where we used the fact that $\Psi_{\!f}$ is fiber-preserving, i.e., $\pi_{Q} \circ \Psi_{\!f} = \pi_{Q}$.
  Hence we have $\Psi_{\!f}^{*}\bar{\Theta} = f\bar{\Theta}$, and thus
  \begin{align}
    \Psi_{\!f}^{*}\bar{\Omega}
    &= \Psi_{\!f}^{*}(-d\bar{\Theta})
    \nonumber\\
    &= -d(\Psi_{\!f}^{*}\bar{\Theta})
    \nonumber\\
    &= -d( f \bar{\Theta} )
    \nonumber\\
    &= -df \wedge \bar{\Theta} - f d\bar{\Theta}
    \nonumber\\
    &= f \bar{\Omega} - df \wedge \bar{\Theta}.
    \label{eq:Psi_f^*barOmega}
  \end{align}
  Therefore, using the fact that $\alpha \wedge \beta = \beta \wedge \alpha$ for any two-forms $\alpha$ and $\beta$, we have
  \begin{align*}
    (\Psi_{\!f}^{*}\bar{\Omega}) \wedge \dots \wedge (\Psi_{\!f}^{*}\bar{\Omega})
    &= f^{\bar{n}}\bar{\Omega} \wedge \dots \wedge \bar{\Omega}
    \\
    &\quad + \sum_{k=1}^{\bar{n}} {\bar{n} \choose k} (-1)^{k} f^{\bar{n}-k} \underbrace{\bar{\Omega} \wedge \dots \wedge \bar{\Omega}}_{\bar{n}-k} \wedge \underbrace{ (df \wedge \bar{\Theta}) \wedge \dots \wedge (df \wedge \bar{\Theta}) }_{k}.
  \end{align*}
  Let us show that the second term vanishes.
  Since $f$ is constant on fibers, we have
  \begin{equation*}
    df = \pd{f}{q^{a}}\,dq^{a}.
  \end{equation*}
  Therefore,
  \begin{equation*}
    df \wedge \bar{\Theta} = p_{b} \pd{f}{q^{a}}\,dq^{a} \wedge dq^{b}
  \end{equation*}
  and thus $df \wedge \bar{\Theta}$ does not contain any term with $dp_{a}$'s.
  On the other hand, $\underbrace{\bar{\Omega} \wedge \dots \wedge \bar{\Omega}}_{\bar{n}-k}$ contains only $\bar{n}-k$ of $dp_{a}$'s.
  Therefore, the $2\bar{n}$-form
  \begin{equation*}
    \underbrace{\bar{\Omega} \wedge \dots \wedge \bar{\Omega}}_{\bar{n}-k} \wedge \underbrace{ (df \wedge \bar{\Theta}) \wedge \dots \wedge (df \wedge \bar{\Theta}) }_{k}
  \end{equation*}
  contains only $\bar{n}-k$ of $dp_{a}$'s, and thus $\bar{n} + k$ of $dq^{a}$'s, which implies that this $2\bar{n}$-form must vanish.
\end{proof}

\begin{definition}
  Let $M$ be an $n$-dimensional orientable manifold, and $\mu$ be a volume form, i.e., a nowhere-vanishing $n$-form.
  Then, the divergence $\divergence_{\mu}(X)$ of a vector field $X$ on $M$ relative to $\mu$ is defined by
  \begin{equation}
    \label{eq:div-def}
    \pounds_{X} \mu = \divergence_{\mu}(X)\,\mu.
  \end{equation}
  Therefore, the flow of $X$ is volume-preserving if and only if $\divergence_{\mu}(X) = 0$.
\end{definition}

\begin{lemma}
  Let $M$ be an orientable differentiable manifold with a volume form $\mu$, $X$ a vector field on $M$, and $f$ a nowhere-vanishing smooth function on $M$.
  Then, the following identity holds:
  \begin{equation}
    \label{eq:div_identity}
    \divergence_{\mu}(f X) = f \divergence_{f \mu}(X).
  \end{equation}
\end{lemma}

\begin{proof}
  We have the identities~\citep[see, e.g.,][Proposition~2.5.23 on p.~130]{AbMa1978}
  \begin{equation*}
    \divergence_{f \mu}(X) = \divergence_{\mu}(X) + \frac{1}{f} X[f],
    \qquad
    \divergence_{\mu}(f X) = f \divergence_{\mu}(X) + X[f].
  \end{equation*}
  Multiplying the first by $f$ and taking the difference of both sides, we have the desired identity.
\end{proof}

\begin{proof}[Proof of Theorem~\ref{thm:Generalized_FeJo2004}]
   As shown in Eq.~\eqref{eq:Psi_f-related}, the vector fields $\bar{X}/f$ and $\bar{X}_{\rm C}$ are $\Psi_{\!f}$-related.
   Therefore,
   \begin{align*}
    \pounds_{\bar{X}/f}(\Psi_{\!f}^{*} \bar{\Omega}) = \Psi_{\!f}^{*} \pounds_{\bar{X}_{\rm C}} \bar{\Omega} = 0,
   \end{align*}
   since $\bar{X}_{\rm C}$ is assumed to be symplectic; thus
   \begin{equation*}
     \pounds_{\bar{X}/f} [ \underbrace{(\Psi_{\!f}^{*}\bar{\Omega}) \wedge \dots \wedge (\Psi_{\!f}^{*}\bar{\Omega})}_{\bar{n}} ] = 0.
   \end{equation*}
   However, by Lemma~\ref{lem:f_Omega}, we have
   \begin{equation*}
     \pounds_{\bar{X}/f} (f^{\bar{n}} \underbrace{\bar{\Omega} \wedge \dots \wedge \bar{\Omega}}_{\bar{n}} ) = 0,
   \end{equation*}
   and hence $\pounds_{\bar{X}/f}(f^{\bar{n}} \bar{\Lambda}) = 0$; this implies $\divergence_{f^{\bar{n}}\bar{\Lambda}}(\bar{X}/f) = 0$.
   Then, the above lemma gives
   \begin{equation*}
     \divergence_{f^{\bar{n}-1}\bar{\Lambda}}(\bar{X})
     = f \divergence_{f^{\bar{n}}\bar{\Lambda}}(\bar{X}/f)
     = 0,
   \end{equation*}
   which implies $\pounds_{\bar{X}}(f^{\bar{n}-1}\bar{\Lambda}) = 0$.
\end{proof}

\subsection{The Chaplygin Hamiltonization}\label{chsc2}
Here we discuss the so-called {\em Chaplygin Hamiltonization} of the reduced system, Eq.~\eqref{eq:ReducedChaplyginSystem}.
Let us first find the equation satisfied by the vector field $\bar{X}_{\rm C}$ defined in Eq.~\eqref{eq:barX_C}.
\begin{lemma}
  \label{lem:barX_C-eq}
  The vector field $\bar{X}_{\rm C} \in \mathfrak{X}(T^{*}\bar{Q})$ satisfies the following equation:
  \begin{equation}
    \label{eq:barX_C-eq}
    i_{\bar{X}_{\rm C}}
    \brackets{
      \bar{\Omega} + \frac{1}{f} \parentheses{ df \wedge \bar{\Theta} - f\,\Xi }
    }
    = d\bar{H}_{\rm C},
  \end{equation}
  where $\bar{H}_{\rm C}: T^{*}\bar{Q} \to \R$ is defined by
  \begin{equation}
    \label{eq:ChaplyginHamiltonian}
    \bar{H}_{\rm C} \defeq \bar{H} \circ \Psi_{1/f}.
  \end{equation}
\end{lemma}

\begin{proof}
  As shown in Eq.~\eqref{eq:Psi_f-related}, the vector fields $\bar{X}/f$ and $\bar{X}_{\rm C}$ are $\Psi_{\!f}$-related.
  Therefore, $\Psi_{\!f}^{*} i_{\bar{X}_{\rm C}} \alpha = i_{\bar{X}/f} \Psi_{\!f}^{*} \alpha$ for any differential form $\alpha$~\citep[see, e.g.,][Proposition~2.4.14]{AbMa1978}; in particular, for $\alpha = \bar{\Omega}$, we have
  \begin{equation*}
    \Psi_{\!f}^{*} i_{\bar{X}_{\rm C}} \bar{\Omega} = i_{\bar{X}/f} \Psi_{\!f}^{*} \bar{\Omega}.
  \end{equation*}
  Using Eqs.~\eqref{eq:Psi_f^*barOmega} and \eqref{eq:ReducedChaplyginSystem} on the right-hand side, we have
  \begin{align*}
    i_{\bar{X}/f} \Psi_{\!f}^{*} \bar{\Omega}
    &= i_{\bar{X}/f} (f \bar{\Omega} - df \wedge \bar{\Theta})
    \\
    &= i_{\bar{X}} \bar{\Omega} - i_{\bar{X}/f} \parentheses{ df \wedge \bar{\Theta} }
    \\
    &= d\bar{H} + i_{\bar{X}}\Xi - i_{\bar{X}/f}\parentheses{ df \wedge \bar{\Theta} }
    \\
    &= d\bar{H} - i_{\bar{X}/f}\parentheses{ df \wedge \bar{\Theta} - f\,\Xi }.
  \end{align*}
  Therefore,
  \begin{equation*}
    \Psi_{\!f}^{*} i_{\bar{X}_{\rm C}} \bar{\Omega} + i_{\bar{X}/f}\parentheses{ df \wedge \bar{\Theta} - f\,\Xi } = d\bar{H},
  \end{equation*}
  and then applying $\Psi_{1/f}^{*}$ to both sides gives
  \begin{equation*}
    i_{\bar{X}_{\rm C}} \bar{\Omega} + \Psi_{1/f}^{*}\, i_{\bar{X}/f}\parentheses{ df \wedge \bar{\Theta} - f\,\Xi } = d\bar{H}_{\rm C}.
  \end{equation*}
  Since the vector fields $\bar{X}_{\rm C}$ and $\bar{X}/f$ are $\Psi_{1/f}$-related, we have $\Psi_{1/f}^{*} i_{\bar{X}/f} \alpha = i_{\bar{X}_{\rm C}} \Psi_{1/f}^{*} \alpha$ for any differential form $\alpha$; hence
  \begin{align*}
    \Psi_{1/f}^{*}\, i_{\bar{X}/f} \parentheses{ df \wedge \bar{\Theta} - f\,\Xi }
    &= i_{\bar{X}_{\rm C}} \Psi_{1/f}^{*} \parentheses{ df \wedge \bar{\Theta} - f\,\Xi }
    \\
    &= i_{\bar{X}_{\rm C}}\brackets{ d(\Psi_{1/f}^{*} f) \wedge (\Psi_{1/f}^{*}\bar{\Theta}) - \Psi_{1/f}^{*}(f\,\Xi) }
    \\
    &= i_{\bar{X}_{\rm C}}\brackets{ df \wedge (\bar{\Theta}/f) - \Xi }
    \\
    &= i_{\bar{X}_{\rm C}} \brackets{ \frac{1}{f}\, \parentheses{ df \wedge \bar{\Theta} - f\,\Xi } },
  \end{align*}
  where $\Psi_{1/f}^{*} f = f$ since $f$ is constant on each fiber; $\Psi_{1/f}^{*}\bar{\Theta} = \bar{\Theta}/f$ as in the proof of Lemma~\ref{lem:f_Omega}; $\Psi_{1/f}^{*}\Xi = \Xi/f$ follows from the following calculation: From the definition of $\Xi$ in Eq.~\eqref{eq:Xi-def}, we have
  \begin{align*}
    (\Psi_{1/f}^{*} \Xi)_{\alpha_{\bar{q}}}(\mathcal{Y}_{\alpha_{\bar{q}}}, \mathcal{Z}_{\alpha_{\bar{q}}})
    &=
    \Xi_{\alpha_{\bar{q}}/f}\parentheses{ T\Psi_{1/f}(\mathcal{Y}_{\alpha_{\bar{q}}}), T\Psi_{1/f}(\mathcal{Z}_{\alpha_{\bar{q}}}) }
    \\
    &= \ip{{\bf J} \circ \hl^{\mathcal{M}}_{q} (\alpha_{\bar{q}}/f)}{ \mathcal{B}_{q}\parentheses{ \hl^{\mathcal{D}}_{q} (Y_{\bar{q}}), \hl^{\mathcal{D}}_{q}(Z_{\bar{q}}) } }
    \\
    &= \frac{1}{f(\bar{q})} \ip{{\bf J} \circ \hl^{\mathcal{M}}_{q} (\alpha_{\bar{q}})}{ \mathcal{B}_{q}\parentheses{ \hl^{\mathcal{D}}_{q} (Y_{\bar{q}}), \hl^{\mathcal{D}}_{q}(Z_{\bar{q}}) } }
    \\
    &= \frac{1}{f(\bar{q})}\, \Xi_{\alpha_{\bar{q}}}(\mathcal{Y}_{\alpha_{\bar{q}}}, \mathcal{Z}_{\alpha_{\bar{q}}}),
  \end{align*}
  where, in the second line, we defined $Y_{\bar{q}}, Z_{\bar{q}} \in T_{\bar{q}}\bar{Q}$ as
  \begin{equation*}
    Y_{\bar{q}} \defeq T\pi_{\bar{Q}} \circ T\Psi_{1/f}(\mathcal{Y}_{\alpha_{\bar{q}}})
    = T(\pi_{\bar{Q}} \circ \Psi_{1/f})(\mathcal{Y}_{\alpha_{\bar{q}}})
    = T\pi_{\bar{Q}}(\mathcal{Y}_{\alpha_{\bar{q}}}),
  \end{equation*}
  and $Z_{\bar{q}}$ in the same way, which coincide the ones introduced earlier when defining $\Xi$; the third line follows from the linearity of $\hl^{\mathcal{M}}$ and also of ${\bf J}$ in the fiber variables.
\end{proof}

\begin{proposition}[Necessary and Sufficient Condition for Hamiltonization]
  \label{prop:NSCondition-Hamiltonization}
  The vector field $\bar{X}_{\rm C} \in \mathfrak{X}(T^{*}\bar{Q})$ satisfies Hamilton's equations
  \begin{equation}
    \label{eq:barX_C-HamiltonEq-0}
    i_{\bar{X}_{\rm C}} \bar{\Omega} = d\bar{H}_{\rm C}
  \end{equation}
   if and only if the one-form $i_{\bar{X}_{\rm C}}\parentheses{ df \wedge \bar{\Theta} - f\,\Xi }$ vanishes.
\end{proposition}

\begin{proof}
  Follows immediately from Lemma~\ref{lem:barX_C-eq}.
\end{proof}

\begin{remark}
  Locally, the above necessary and sufficient condition is precisely Eq.~(2.17) in \citet{FeMeBl2009}.
\end{remark}

\begin{definition}
  The process of finding an $f$ satisfying the above condition is called {\em Chaplygin Hamiltonization}, or just {\em Hamiltonization} for short; the resulting Hamiltonian system, Eq.~\eqref{eq:barX_C-HamiltonEq-0}, is called the {\em Hamiltonized system}; we would like to call $\bar{H}_{\rm C}$ a {\em Chaplygin Hamiltonian}.
\end{definition}

Now, combining Proposition~\ref{prop:NSCondition-Hamiltonization} with Theorem~\ref{thm:Generalized_FeJo2004} or Corollary~\ref{cor:FeJo2004}, we have
\begin{corollary}
 Suppose there exists a nowhere-vanishing fiber-wise constant function $f: T^{*}\bar{Q} \to \R$ such that $i_{\bar{X}_{\rm C}}\parentheses{ df \wedge \bar{\Theta} - f\,\Xi }$ vanishes.
 Then, the $2\bar{n}$-form $f^{\bar{n}-1} \bar{\Lambda}$ is an invariant measure of the reduced system, Eq.~\eqref{eq:ReducedChaplyginSystem}.
\end{corollary}

We now state the main result of this section.
The following theorem will be used in the next section in relation to the nonholonomic Hamilton--Jacobi theory:
\begin{theorem}[A Sufficient Condition for Hamiltonization]
  \label{thm:SufficientCondition}
  Suppose there exists a nowhere-vanishing fiber-wise constant function $f: T^{*}\bar{Q} \to \R$ that satisfies the equation
  \begin{equation}
    \label{eq:f}
    df \wedge \bar{\Theta} = f\,\Xi.
  \end{equation}
  Then, the vector field $\bar{X}_{\rm C} \in \mathfrak{X}(T^{*}\bar{Q})$ satisfies the following Hamilton's equations:
  \begin{equation}
    \label{eq:barX_C-HamiltonEq}
    i_{\bar{X}_{\rm C}} \bar{\Omega} = d\bar{H}_{\rm C},
  \end{equation}
  and, as a result, the reduced nonholonomic dynamics Eq.~\eqref{eq:ReducedChaplyginSystem} has the invariant measure $f^{\bar{n}-1} \bar{\Lambda}$.
\end{theorem}

\begin{proof}
  Straightforward from Lemma~\ref{lem:barX_C-eq} and Corollary~\ref{cor:FeJo2004}.
\end{proof}

\begin{remark}
  Locally, the sufficient condition ~\eqref{eq:f} becomes condition
  (2.22) in \citet{FeMeBl2009}.
\end{remark}

\begin{remark}
  As shown by \citet{St1989} (see also \citet{CaCoLeMa2002}), Eq.~\eqref{eq:f} is also a sufficient condition for the two-form $\bar{\Omega}_{f} \defeq f (\bar{\Omega} - \Xi)$ to be closed, so that Eq.~\eqref{eq:ReducedChaplyginSystem} becomes
  \begin{equation*}
    i_{\bar{X}/f} \bar{\Omega}_{f} = d\bar{H},
  \end{equation*}
  and so the dynamics of $\bar{X}/f$ is Hamiltonian with the non-standard symplectic form $\bar{\Omega}_{f}$.
\end{remark}


\section{Nonholonomic Hamilton--Jacobi Theory via Chaplygin Hamiltonization}
\label{sec:NHHJ}
\subsection{The Chaplygin Hamilton--Jacobi Equation}
Since the Hamiltonized system, Eq.~\eqref{eq:barX_C-HamiltonEq}, is a canonical Hamiltonian system on $T^{*}\bar{Q}$, we may apply the conventional Hamilton--Jacobi theory (see, e.g., \citet[][Chapter~5]{AbMa1978}) to the system and obtain the (time-independent) Hamilton--Jacobi equation:
\begin{equation}
\label{eq:ChaplyginHJ}
  \bar{H}_{\rm C} \circ d\bar{W} = E,
\end{equation}
with an unknown function $\bar{W}: \bar{Q} \to \R$ and a constant $E$ (the total energy).
We would like to call Eq.~\eqref{eq:ChaplyginHJ} the {\em Chaplygin Hamilton--Jacobi equation}.

Now that we have two Hamilton--Jacobi equations for Chaplygin systems, i.e., the nonholonomic Hamilton--Jacobi equation~\eqref{eq:NHHJ} and the Chaplygin Hamilton--Jacobi equation~\eqref{eq:ChaplyginHJ}, a natural question to ask is: What is the relationship between the two?

\subsection{Relationship between the Chaplygin H--J and Nonholonomic H--J Equations}\label{hjnhsec}
In relating the Chaplygin Hamilton--Jacobi equation~\eqref{eq:ChaplyginHJ} to the nonholonomic Hamilton--Jacobi equation~\eqref{eq:NHHJ}, a natural starting point is to look into the relationship between the Chaplygin Hamiltonian $\bar{H}_{\rm C}$ and the original Hamiltonian $H$ (recall from Eqs.~\eqref{eq:barH} and \eqref{eq:ChaplyginHamiltonian} that they are related through the Hamiltonian $\bar{H}$); the upper half of the following commutative diagram shows their relationship.
\begin{equation}
  \label{dia:gamma-dbarW}
  \vcenter{
    \xymatrix@!0@R=0.8in@C=0.8in{
      & \R &
      \\
      \mathcal{M} \ar[ru]^{H\!\!} & T^{*}\bar{Q} \ar[u]^{\bar{H}} \ar[l]^{\hl^{\mathcal{M}}} & T^{*}\bar{Q} \ar[ul]_{\!\!\bar{H}_{\rm C}} \ar[l]^{\ \Psi_{1/f}}
      \\
      Q \ar[rr]_{\pi} \ar@{-->}[u]^{\gamma} & & \bar{Q} \ar[u]_{d\bar{W}}
    }
  }
\end{equation}
Now, suppose that a function $\bar{W}: \bar{Q} \to \R$ satisfies the Chaplygin Hamilton--Jacobi equation~\eqref{eq:ChaplyginHJ}.
This means that the one-form $d\bar{W}$, seen as a map from $\bar{Q}$ to $T^{*}\bar{Q}$, satisfies $\bar{H}_{\rm C} \circ d\bar{W}(\bar{q}) = E$ for any $\bar{q} \in \bar{Q}$ with some constant $E$; equivalently, $\bar{H}_{\rm C} \circ d\bar{W} \circ \pi (q) = E$ for any $q \in Q$.
The lower half of the above diagram~\eqref{dia:gamma-dbarW} incorporates this view, and also leads us to the following:

\begin{theorem}
  \label{thm:ChaplyginHJ-NHHJ}
  Suppose that there exists a nowhere-vanishing fiber-wise constant function $f: T^{*}\bar{Q} \to \R$ that satisfies Eq.~\eqref{eq:f}, and hence by Theorem~\ref{thm:SufficientCondition}, we have Hamilton's equations \eqref{eq:barX_C-HamiltonEq} for the vector field $\bar{X}_{\rm C}$.
  Let $\bar{W}: \bar{Q} \to \R$ be a solution of the Chaplygin Hamilton--Jacobi equation~\eqref{eq:ChaplyginHJ}, and define $\gamma: Q \to \mathcal{M}$ by
  \begin{equation}
    \label{eq:gamma-dbarW}
    \gamma(q) \defeq \hl^{\mathcal{M}}_{q} \circ \Psi_{1/f} \circ d\bar{W} \circ \pi(q)
    = \hl^{\mathcal{M}}_{q}\parentheses{ \frac{1}{f(\bar{q})} d\bar{W}(\bar{q}) },
  \end{equation}
  where $\bar{q} \defeq \pi(q)$.
  Then $\gamma$ satisfies the nonholonomic Hamilton--Jacobi equation~\eqref{eq:NHHJ} as well as the condition Eq.~\eqref{eq:dgamma}.
\end{theorem}

\begin{remark}
  Notice that Theorem~\ref{thm:ChaplyginHJ-NHHJ} relates a solution of the Chaplygin Hamilton--Jacobi equation, which is for the {\em reduced} dynamics defined by Eq.~\eqref{eq:barX_C-HamiltonEq}, with that of the nonholonomic Hamilton--Jacobi equation for the {\em full} dynamics defined by Eq.~\eqref{eq:NHHam}.
  Therefore, the theorem provides a method to integrate the full dynamics by solving a Hamilton--Jacobi equation for the reduced dynamics.
\end{remark}

\begin{proof}
  That the one-form $\gamma$ defined by Eq.~\eqref{eq:gamma-dbarW} satisfies the nonholonomic Hamilton--Jacobi equation~\eqref{eq:NHHJ} follows from the diagram \eqref{dia:gamma-dbarW}.
  To show that it also satisfies the condition Eq.~\eqref{eq:dgamma}, we perform the following lengthy calculations:
  Let $Y^{\rm h}, Z^{\rm h} \in \mathfrak{X}(Q)$ be arbitrary horizontal vector fields, i.e., $Y^{\rm h}_{q}, Z^{\rm h}_{q} \in \mathcal{D}_{q}$ for any $q \in Q$.
  We start from the following identity:
  \begin{equation}
    \label{eq:gamma-identity}
    d\gamma(Y^{\rm h}, Z^{\rm h})
    = Y^{\rm h}[\gamma(Z^{\rm h})] - Z^{\rm h}[\gamma(Y^{\rm h})] - \gamma([Y^{\rm h},Z^{\rm h}]).
  \end{equation}
  The goal is to show that the right-hand side vanishes.
  Let us first evaluate the first two terms on the right-hand side of the above identity at an arbitrary point $q \in Q$:
  Let $Z_{\bar{q}} \defeq T_{q}\pi_{Q}(Z^{\rm h}_{q}) \in T_{\bar{q}}\bar{Q}$, then $Z^{\rm h}_{q} = \hl^{\mathcal{D}}_{q}(Z_{\bar{q}})$.
  Thus, using Lemma~\ref{lem:hl-pairing}, we have\footnote{Recall that $f: T^{*}\bar{Q} \to \R$ is fiber-wise constant and thus, with a slight abuse of notation, we may write $f(\alpha_{\bar{q}}) = f(\bar{q})$ for any $\alpha_{\bar{q}} \in T^{*}_{\bar{q}}\bar{Q}$; therefore $f$ may be seen as a function on $\bar{Q}$ as well.}
  \begin{align*}
    \gamma(Z^{\rm h})(q)
    &= \ip{ \hl^{\mathcal{M}}_{q} \circ \Psi_{1/f} \circ d\bar{W}(\bar{q}) }{ \hl^{\mathcal{D}}_{q}(Z_{\bar{q}}) }
    \\
    &= \ip{ \Psi_{1/f} \circ d\bar{W}(\bar{q}) }{ Z_{\bar{q}} }
    \\
    &= \frac{1}{f(\bar{q})}\, d\bar{W}(Z)(\bar{q}).
  \end{align*}
  Hence, defining a function $\gamma_{Z}: \bar{Q} \to \R$ by
  \begin{equation*}
    \gamma_{Z}(\bar{q}) \defeq \frac{1}{f(\bar{q})}\, d\bar{W}(Z)(\bar{q}),
  \end{equation*}
  we have $\gamma(Z^{\rm h}) = \gamma_{Z} \circ \pi$.
  Therefore, defining $Y_{\bar{q}} \defeq T_{q}\pi(Y^{\rm h}_{q})$, i.e., $Y^{\rm h}_{q} = \hl^{\mathcal{D}}_{q}(Y_{\bar{q}})$,
  \begin{align*}
    Y^{\rm h}[ \gamma(Z^{\rm h}) ](q)
    &= Y^{\rm h}[ \gamma_{Z} \circ \pi ](q)
    \\
    &= \ip{ d(\gamma_{Z} \circ \pi)_{q} }{ Y^{\rm h}_{q} }
    \\
    &= \ip{ (\pi^{*} d\gamma_{Z})_{q} }{ Y^{\rm h}_{q} }
    \\
    &= \ip{ d\gamma_{Z}(\bar{q}) }{ T_{q}\pi(Y^{\rm h}_{q}) }
    \\
    &= \ip{ d\gamma_{Z}(\bar{q}) }{ Y_{\bar{q}} }
    \\
    &= Y[\gamma_{Z}](\bar{q})
    \\
    &= Y\brackets{ \frac{1}{f}\, d\bar{W}(Z) }(\bar{q})
    \\
    &= \parentheses{ 
      \frac{1}{f}\,Y\brackets{ Z\brackets{\bar{W}} } - \frac{1}{f^{2}}\, df(Y)\, d\bar{W}(Z)
      }(\bar{q}).
  \end{align*}
  Hence we have
  \begin{align}
     Y^{\rm h}[ \gamma(Z^{\rm h}) ] -  Z^{\rm h}[ \gamma(Y^{\rm h}) ]
     &= \frac{1}{f} \parentheses{ 
       Y \brackets{ Z\brackets{\bar{W}} } - Z \brackets{ Y\brackets{\bar{W}} }
     }
     - \frac{1}{f^{2}} \parentheses{
       df(Y)\, d\bar{W}(Z)
       - df(Z)\, d\bar{W}(Y)
    }
    \nonumber\\
    &= \frac{1}{f}\, d\bar{W}([Y,Z])
    - \frac{1}{f^{2}}\, df \wedge d\bar{W}(Y, Z),
    \label{eq:YgammaZ-ZgammaY}
  \end{align}
  where we have omitted $q$ and $\bar{q}$ for simplicity.

  Now, let us evaluate the last term on the right-hand side of Eq.~\eqref{eq:gamma-identity}:
  First we would like to decompose $[Y^{\rm h},Z^{\rm h}]_{q}$ into the horizontal and vertical part.
  Since both $Y^{\rm h}$ and $Z^{\rm h}$ are horizontal, we have\footnote{See, e.g., \citet[][Proposition~1.3~(3), p.~65]{KoNo1963}.}
  \begin{equation*}
    \hor([Y^{\rm h},Z^{\rm h}]_{q}) = \hl^{\mathcal{D}}_{q}([Y, Z]_{\bar{q}}),
  \end{equation*}
  whereas the vertical part is
  \begin{equation*}
    \ver([Y^{\rm h},Z^{\rm h}]_{q}) = \parentheses{ \mathcal{A}_{q}([Y^{\rm h},Z^{\rm h}]_{q}) }_{Q}(q) = -\parentheses{ \mathcal{B}_{q}(Y^{\rm h}_{q},Z^{\rm h}_{q}) }_{Q}(q),
  \end{equation*}
  where we used the following relation between the connection $\mathcal{A}$ and its curvature $\mathcal{B}$ that hold for horizontal vector fields $Y^{\rm h}$ and $Z^{\rm h}$:
  \begin{align*}
    \mathcal{B}_{q}(Y^{\rm h}_{q}, Z^{\rm h}_{q}) &= d\mathcal{A}_{q}(Y^{\rm h}_{q}, Z^{\rm h}_{q})
    \\
    &= Y^{\rm h}[ \mathcal{A}(Z^{\rm h}) ](q) - Y^{\rm h}[ \mathcal{A}(Z^{\rm h}) ](q) - \mathcal{A}([Y^{\rm h},Z^{\rm h}])(q)
    \\
    &= -\mathcal{A}([Y^{\rm h},Z^{\rm h}])(q).
  \end{align*}
  As a result, we have the decomposition
  \begin{equation*}
    [Y^{\rm h},Z^{\rm h}]_{q} = \hl^{\mathcal{D}}_{q}([Y, Z]_{\bar{q}}) - \parentheses{ \mathcal{B}_{q}(Y^{\rm h}_{q},Z^{\rm h}_{q}) }_{Q}(q).
  \end{equation*}
  Therefore,
  \begin{align}
    \gamma([Y^{\rm h},Z^{\rm h}])(q)
    &= \ip{ \hl^{\mathcal{M}}_{q} \circ \Psi_{1/f} \circ d\bar{W} \circ \pi(q) }{ \hl^{\mathcal{D}}_{q}([Y, Z]_{\bar{q}}) }
    \nonumber\\
    &\quad - \ip{ \hl^{\mathcal{M}}_{q} \circ \Psi_{1/f} \circ d\bar{W} \circ \pi(q) }{ \parentheses{ \mathcal{B}_{q}(Y^{\rm h}_{q},Z^{\rm h}_{q}) }_{Q}(q) }
    \nonumber\\
    &= \ip{ \Psi_{1/f} \circ d\bar{W}(\bar{q}) }{ [Y, Z]_{\bar{q}} }
    - \ip{ {\bf J}\parentheses{ \hl^{\mathcal{M}}_{q} \circ \Psi_{1/f} \circ d\bar{W}(\bar{q}) } }{ \mathcal{B}_{q}(Y^{\rm h}_{q},Z^{\rm h}_{q}) }
    \nonumber\\
    &= \frac{1}{f(\bar{q})} \ip{ d\bar{W}(\bar{q}) }{ [Y, Z]_{\bar{q}} }
    - \ip{ {\bf J} \circ \hl^{\mathcal{M}}_{q} \parentheses{ d\bar{W}(\bar{q})/f(\bar{q}) } }{ \mathcal{B}_{q}(Y^{\rm h}_{q},Z^{\rm h}_{q}) }
    \nonumber\\
    &= \frac{1}{f(\bar{q})}\, d\bar{W}([Y,Z])(\bar{q})
    - \frac{1}{f(\bar{q})} \ip{ {\bf J} \circ \hl^{\mathcal{M}}_{q} \parentheses{ d\bar{W}(\bar{q})} }{ \mathcal{B}_{q}\parentheses{\hl^{\mathcal{D}}_{q}(Y_{\bar{q}}), \hl^{\mathcal{D}}_{q}(Z_{\bar{q}})} }
    \nonumber\\
    &= \frac{1}{f(\bar{q})}\, d\bar{W}([Y,Z])(\bar{q})
    - \frac{1}{f(\bar{q})} (d\bar{W})^{*}\Xi(Y, Z)(\bar{q}),
    \label{eq:gammaYZ}
  \end{align}
  where the second equality follows from Lemma~\ref{lem:hl-pairing} and the definition of the momentum map ${\bf J}$; the fourth one follows from the linearity of $\hl^{\mathcal{M}}$ and also of ${\bf J}$ in the fiber variables; the last one follows from the definition of $\Xi$ in Eq.~\eqref{eq:Xi-def}:
  Since $\pi_{\bar{Q}} \circ d\bar{W} = \id_{\bar{Q}}$ and thus $T\pi_{\bar{Q}} \circ Td\bar{W} = \id_{T\bar{Q}}$, we have
  \begin{align*}
    (d\bar{W})^{*}\Xi(Y, Z)(\bar{q})
    &= \Xi_{d\bar{W}(\bar{q})} \parentheses{ Td\bar{W}(Y_{\bar{q}}), Td\bar{W}(Z_{\bar{q}}) }
    \\
    &= \ip{ {\bf J} \circ \hl^{\mathcal{M}}_{q} \parentheses{ d\bar{W}(\bar{q})} }{ \mathcal{B}_{q}\parentheses{\hl^{\mathcal{D}}_{q}(Y_{\bar{q}}), \hl^{\mathcal{D}}_{q}(Z_{\bar{q}})} }.
  \end{align*}
  Substituting Eqs.~\eqref{eq:YgammaZ-ZgammaY} and \eqref{eq:gammaYZ} into Eq.~\eqref{eq:gamma-identity}, we obtain
  \begin{align*}
    d\gamma(Y^{\rm h}, Z^{\rm h})
    &= -\frac{1}{f^{2}}\, df \wedge d\bar{W}(Y, Z) + \frac{1}{f} (d\bar{W})^{*}\Xi(Y, Z)
    \\
    &= -\frac{1}{f^{2}} \parentheses{ df \wedge d\bar{W} - f\,(d\bar{W})^{*}\Xi }(Y, Z)
    \\
    &= -\frac{1}{f^{2}} (d\bar{W})^{*} \parentheses{ df \wedge \bar{\Theta} - f\,\Xi }(Y, Z)
    \\
    &= 0,
  \end{align*}
  where the third line follows since\footnote{Again recall that $f: T^{*}\bar{Q} \to \R$ may be seen as a function on $\bar{Q}$ as well.} $(d\bar{W})^{*}f(\bar{q}) = f\parentheses{ d\bar{W}(\bar{q}) } = f(\bar{q})$ and also that $(d\bar{W})^{*} \bar{\Theta} = d\bar{W}$~\citep[see, e.g.,][Proposition~3.2.11 on p.~179]{AbMa1978}; the last line follows from Eq.~\eqref{eq:f}, which is assumed to be satisfied.
\end{proof}


\section{Examples}\label{exsec}
\begin{example}[The vertical rolling disk; see, e.g., \citet{Bl2003}]
  \label{ex:VRD}
  Consider the motion of the vertical rolling disk of radius $R$ shown in Fig.~\ref{fig:VRD}.
  \begin{figure}[htbp]
    \centering
    \includegraphics[width=.65\linewidth]{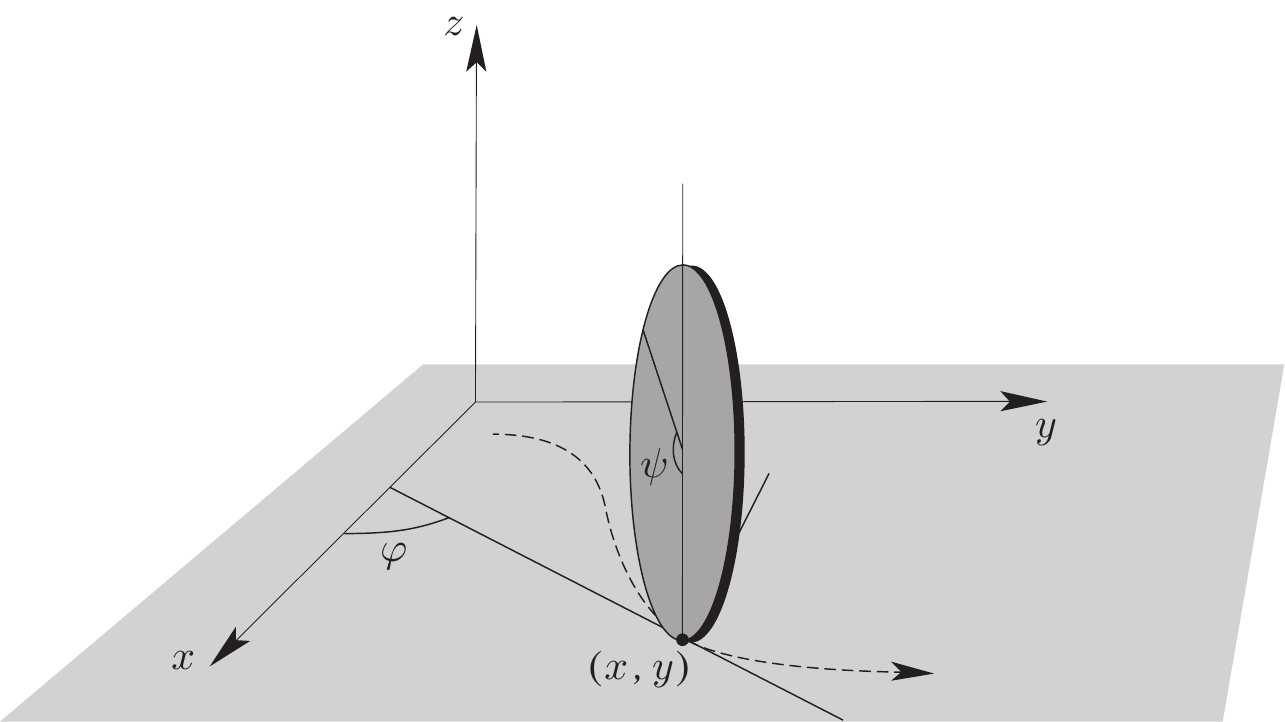}
    \caption{Vertical rolling disk.}
    \label{fig:VRD}
  \end{figure}
  The configuration space is
  \begin{equation*}
    Q = SE(2) \times \mathbb{S}^{1} = (SO(2) \ltimes \R^{2}) \times \mathbb{S}^{1} = \{ (\varphi, x, y, \psi) \}.
  \end{equation*}
  Suppose that $m$ is the mass of the disk, $I$ is the moment of inertia of the disk about the axis perpendicular to the plane of the disk, and $J$ is the moment of inertia about an axis in the plane of the disk (both axes passing through the disk's center).
  The Lagrangian $L: TQ \to \R$ and the Hamiltonian $H: T^{*}Q \to \R$ are given by
  \begin{equation*}
    L = \frac{1}{2} m \parentheses{ \dot{x}^{2} + \dot{y}^{2} } + \frac{1}{2} J \dot{\varphi}^{2} + \frac{1}{2}I\dot{\psi}^{2}
  \end{equation*}
  and
  \begin{equation*}
    H = \frac{1}{2}\parentheses{ \frac{p_{x}^{2} + p_{y}^{2}}{m} + \frac{p_{\varphi}^{2}}{J} + \frac{p_{\psi}^{2}}{I} }.
  \end{equation*}
  The velocity constraints are
  \begin{equation*}
    \dot{x} = R\cos\varphi\,\dot{\psi},
    \qquad
    \dot{y} = R\sin\varphi\,\dot{\psi},
  \end{equation*}
  or in terms of constraint one-forms,
  \begin{equation*}
    \omega^{1} = dx - R\cos\varphi\,d\psi,
    \qquad
    \omega^{2} = dy - R\sin\varphi\,d\psi.
  \end{equation*}
  So the constraint distribution $\mathcal{D} \subset TQ$ and the constrained momentum space $\mathcal{M} \subset T^{*}Q$ are given by
  \begin{equation*}
    \mathcal{D}
    = \setdef{ (\dot{\varphi}, \dot{x}, \dot{y}, \dot{\psi}) \in TQ }{ \dot{x} = R\cos\varphi\,\dot{\psi},\ \dot{y} = R\sin\varphi\,\dot{\psi} }
  \end{equation*}
  and
  \begin{equation*}
    \mathcal{M}
    = \setdef{ (p_{\varphi}, p_{x}, p_{y}, p_{\psi}) \in T^{*}Q }{ p_{x} = \frac{m R}{I} \cos\varphi\,p_{\psi},\ p_{y} = \frac{m R}{I} \sin\varphi\,p_{\psi} }.
  \end{equation*}

  Let $G = \R^{2}$ and consider the action of $G$ on $Q$ defined by
  \begin{equation*}
    G \times Q \to Q;
    \quad
    \parentheses{(a, b), (\varphi, x, y, \psi)} \mapsto (\varphi, x + a, y + b, \psi).
  \end{equation*}
  Then, the system is a Chaplygin system in the sense of Definition~\ref{def:ChaplyginSystems}.
  The Lie algebra $\mathfrak{g}$ is identified with $\R^{2}$ in this case; let us use $(\xi, \eta)$ as the coordinates for $\mathfrak{g}$.
  Then, we may write the connection $\mathcal{A}: TQ \to \mathfrak{g}$ as
  \begin{equation}
    \label{eq:mathcalA-VRD}
    \mathcal{A} = (dx - R\cos\varphi\,d\psi) \otimes \pd{}{\xi} + (dy - R\sin\varphi\,d\psi) \otimes \pd{}{\eta},
  \end{equation}
  and hence its curvature as
  \begin{equation}
    \label{eq:mathcalB-VRD}
    \mathcal{B} = R \parentheses{ \sin\varphi\,d\varphi \wedge d\psi \otimes \pd{}{\xi} - \cos\varphi\,d\varphi \wedge d\psi \otimes \pd{}{\eta} }.
  \end{equation}
  Furthermore, the momentum map ${\bf J}: T^{*}Q \to \mathfrak{g}^{*}$ is given by
  \begin{equation}
    \label{eq:J-VRD}
    {\bf J}(p_{q}) = p_{x}\,d\xi + p_{y}\,d\eta.
  \end{equation}

  The quotient space is $\bar{Q} \defeq Q/G = \{(\varphi, \psi)\}$.
  The reduced Hamiltonian $\bar{H}: T^{*}\bar{Q} \to \R$ is
  \begin{equation}
  	\label{eq:ChapHamVD}
    \bar{H} = \frac{1}{2}\parentheses{ \frac{1}{J}\,p_{\varphi}^{2} + \frac{I + m R^{2}}{I^{2}}\,p_{\psi}^{2} }.
  \end{equation}
  A simple calculation shows that the horizontal lift $\hl^{\mathcal{M}}: T^{*}\bar{Q} \to \mathcal{M}$ is given by
  \begin{equation}
    \label{eq:hl^M-VRD}
    \hl^{\mathcal{M}}(p_{\varphi}, p_{\psi})
    = \parentheses{ p_{\varphi},\, \frac{m R}{I} \cos\varphi\,p_{\psi},\, \frac{m R}{I} \sin\varphi\,p_{\psi},\, p_{\psi} }.
  \end{equation}
  Then, we find from Eq.~\eqref{eq:Xi-def} along with Eqs.~\eqref{eq:mathcalA-VRD}, \eqref{eq:mathcalB-VRD}, \eqref{eq:J-VRD}, and \eqref{eq:hl^M-VRD} that $\Xi = 0$.
  Therefore, the sufficient condition, Eq.~\eqref{eq:f}, for Chaplygin Hamiltonization reduces to $df \wedge \Theta = 0$, and hence we may choose $f = 1$.
  Thus, the Chaplygin Hamiltonian $\bar{H}_{\rm C}: T^{*}\bar{Q} \to \R$ is identical to $\bar{H}$ (see Eq.~\eqref{eq:ChaplyginHamiltonian}).
  
  To illustrate Theorem \ref{thm:ChaplyginHJ-NHHJ}, we begin with the Chaplygin Hamilton--Jacobi equation~\eqref{eq:ChaplyginHJ}:
  \begin{equation}
    \label{eq:ChaplyginHJ-VRD}
    \frac{1}{2}\brackets{
      \frac{1}{J} \parentheses{ \pd{\bar{W}}{\varphi} }^{2} + \frac{I + m R^{2}}{I^{2}} \parentheses{ \pd{\bar{W}}{\psi} }^{2}
    } = E.
  \end{equation}
  Now, we employ the conventional approach of separation of variables, i.e., assume that $\bar{W}: \bar{Q} \to \R$ takes the following form:
  \begin{equation*}
    \bar{W}(\varphi, \psi) = \bar{W}_{\varphi}(\varphi) + \bar{W}_{\psi}(\psi).
  \end{equation*}
  Then, Eq.~\eqref{eq:ChaplyginHJ-VRD} becomes
  \begin{equation*}
    \frac{1}{2}\brackets{
      \frac{1}{J} \parentheses{ \od{\bar{W}_{\varphi}}{\varphi} }^{2} + \frac{I + m R^{2}}{I^{2}} \parentheses{ \od{\bar{W}_{\psi}}{\psi} }^{2}
    } = E.
  \end{equation*}
  Since the first term on the left-hand side depends only on $\varphi$ and the second only on $\psi$, we obtain the solution
  \begin{equation}
  	\label{eq:VDsimp1}
    \od{\bar{W}_{\varphi}}{\varphi} = \gamma_{\varphi}^{0},
    \qquad
    \od{\bar{W}_{\psi}}{\psi} = \gamma_{\psi}^{0},
  \end{equation}
  where $\gamma_{\varphi}^{0}$ and $\gamma_{\psi}^{0}$ are the constants determined by the initial condition such that
  \begin{equation*}
    \frac{1}{2}\brackets{ \frac{1}{J}(\gamma_{\varphi}^{0})^{2} + \frac{I+m R^{2}}{I^{2}}(\gamma_{\psi}^{0})^{2} } = E.
  \end{equation*}
  Then, Eq.~\eqref{eq:gamma-dbarW} gives
  \begin{equation}
  	\label{eq:VDsoln}
    \gamma(\varphi, x, y, \psi)
    = \gamma_{\varphi}^{0}\,d\varphi+ \frac{m R}{I}\cos\varphi\,\gamma_{\psi}^{0}\,dx + \frac{m R}{I}\sin\varphi\,\gamma_{\psi}^{0}\,dy
    + \gamma_{\psi}^{0}\,d\psi,
  \end{equation}
  which is the solution of the nonholonomic Hamilton--Jacobi equation~\eqref{eq:NHHJ} obtained in \citet[][Example~4.1]{OhBl2009}:
\end{example}

\begin{example}[The knife edge; see, e.g., \citet{Bl2003}]
  \label{ex:KnifeEdge}
  Consider a plane slanted at an angle $\alpha$ from the horizontal and let $(x,y)$ represent the position of the point of contact of the knife edge with respect to a fixed Cartesian coordinate system on the plane~(see Fig.~\ref{fig:KnifeEdge}) and $\varphi$ the angle of it as shown in Fig.~\ref{fig:KnifeEdge}.
  \begin{figure}[htbp]
    \centering
    \includegraphics[width=.65\linewidth]{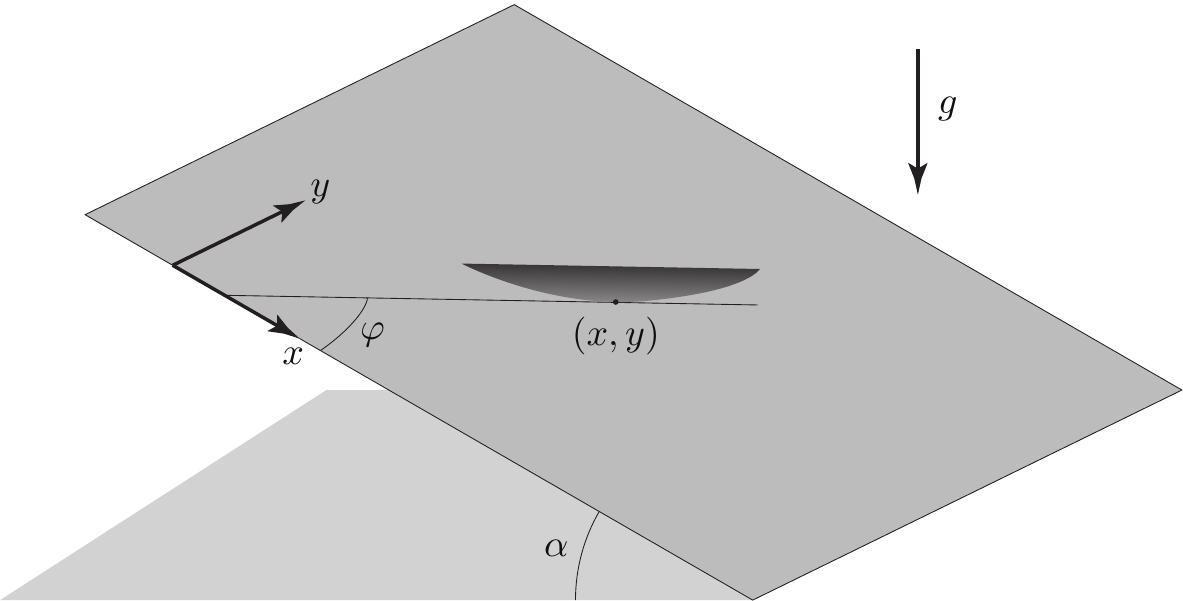}
    \caption{Knife edge on inclined plane.}
    \label{fig:KnifeEdge}
  \end{figure}
  The configuration space is
  \begin{equation*}
    Q = SE(2) = SO(2) \ltimes \R^{2} = \{ (\varphi, x, y) \}.
  \end{equation*}
  Suppose that the mass of the knife edge is $m$, and the moment of inertia about the axis perpendicular to the inclined plane through its contact point is $J$.
  The Lagrangian $L: TQ \to \R$ and the Hamiltonian $H: T^{*}Q \to \R$ are given by
  \begin{equation*}
    L = \frac{1}{2}m \parentheses{\dot{x}^{2} + \dot{y}^{2}} + \frac{1}{2}J \dot{\varphi}^{2} + m g x \sin \alpha
  \end{equation*}
  and
  \begin{equation*}
    H = \frac{1}{2}\parentheses{ \frac{p_{x}^{2} + p_{y}^{2}}{m} + \frac{p_{\varphi}^{2}}{J} } - m g x \sin \alpha.
  \end{equation*}
  The velocity constraint is
  \begin{equation*}
    \sin\varphi\,\dot{x} - \cos\varphi\,\dot{y} = 0
  \end{equation*}
  and so the constraint one-form is
  \begin{equation*}
    \omega^{1} = \sin\varphi\,dx - \cos\varphi\,dy.
  \end{equation*}
  The constraint distribution $\mathcal{D} \subset TQ$ and the constrained momentum space $\mathcal{M} \subset T^{*}Q$ are given by
  \begin{equation*}
    \mathcal{D}
    = \setdef{ (\dot{\varphi}, \dot{x}, \dot{y}) \in TQ }{ \sin\varphi\,\dot{x} - \cos\varphi\,\dot{y} = 0 }
  \end{equation*}
  and
  \begin{equation*}
    \mathcal{M}
    = \setdef{ (p_{\varphi}, p_{x}, p_{y}) \in T^{*}Q }{ \sin\varphi\,p_{x} = \cos\varphi\,p_{y} }.
  \end{equation*}

  Let $G = \R$ and consider the action of $G$ on $Q$ defined by
  \begin{equation*}
    G \times Q \to Q;
    \quad
    \parentheses{a, (\varphi, x, y)} \mapsto (\varphi, x, y + a).
  \end{equation*}
  Then, the system is a Chaplygin system in the sense of Definition~\ref{def:ChaplyginSystems}.
  The Lie algebra $\mathfrak{g}$ is identified with $\R$ in this case; let us use $\eta$ as the coordinate for $\mathfrak{g}$.
  Then, we may write the connection $\mathcal{A}: TQ \to \mathfrak{g}$ as
  \begin{equation}
    \label{eq:mathcalA-KnifeEdge}
    \mathcal{A} = (dy - \tan\varphi\,dx) \otimes \pd{}{\eta},
  \end{equation}
  and hence its curvature as
  \begin{equation}
    \label{eq:mathcalB-KnifeEdge}
    \mathcal{B} = \frac{1}{\cos^{2}\varphi}\,dx \wedge d\varphi \otimes \pd{}{\eta},
  \end{equation}
  where we assume that $\varphi$ stays in the range $(0, \pi/2)$ or $(\pi/2, \pi)$ to avoid singularities.
  Furthermore, the momentum map ${\bf J}: T^{*}Q \to \mathfrak{g}^{*}$ is given by
  \begin{equation}
    \label{eq:J-KnifeEdge}
    {\bf J}(p_{q}) = p_{y}\,d\eta.
  \end{equation}

  The quotient space is $\bar{Q} \defeq Q/G = \{(\varphi, x)\}$.
  The reduced Hamiltonian $\bar{H}: T^{*}\bar{Q} \to \R$ is
  \begin{equation*}
    \bar{H} = \frac{1}{2}\parentheses{ \frac{\cos^{2}\varphi}{m}\,p_{x}^{2} + \frac{1}{J}\,p_{\varphi}^{2} } - m g x \sin \alpha.
  \end{equation*}
  A simple calculation shows that the horizontal lift $\hl^{\mathcal{M}}: T^{*}\bar{Q} \to \mathcal{M}$ is given by
  \begin{equation}
    \label{eq:hl^M-KnifeEdge}
    \hl^{\mathcal{M}}(p_{\varphi}, p_{x})
    = \parentheses{ p_{\varphi}, \cos^{2}\varphi\,p_{x},\, \sin\varphi \cos\varphi\,p_{x} }.
  \end{equation}
  Then, we find from Eq.~\eqref{eq:Xi-def} along with Eqs.~\eqref{eq:mathcalA-KnifeEdge}, \eqref{eq:mathcalB-KnifeEdge}, \eqref{eq:J-KnifeEdge}, and \eqref{eq:hl^M-KnifeEdge} that
  \begin{equation*}
    \Xi = p_{x} \tan\varphi\,dx \wedge d\varphi.
  \end{equation*}
  Therefore, the sufficient condition, Eq.~\eqref{eq:f}, for Chaplygin Hamiltonization gives
  \begin{equation*}
    p_{\varphi} \pd{f}{x} - p_{x} \pd{f}{\varphi} = (p_{x} \tan\varphi)\,f.
  \end{equation*}
  It is easy to find the solution
  \begin{equation}
    \label{eq:f-KnifeEdge}
    f = \cos\varphi.
  \end{equation}
  Note that $f$ is nowhere-vanishing if $\varphi$ is assumed to be in the range $(0, \pi/2)$ or $(\pi/2, \pi)$.

  Then, Eq.~\eqref{eq:ChaplyginHamiltonian} gives the following Chaplygin Hamiltonian:
  \begin{align*}
    \bar{H}_{\rm C}(\varphi, x, p_{\varphi}, p_{x})
    &= \bar{H}\parentheses{ \varphi, x, \frac{p_{\varphi}}{\cos\varphi}, \frac{p_{x}}{\cos\varphi} }
    \\
    &= \frac{1}{2}\parentheses{ \frac{1}{m}\,p_{x}^{2} + \frac{1}{J \cos^{2}\varphi}\,p_{\varphi}^{2} } - m g x \sin \alpha.
  \end{align*}

  The Chaplygin Hamilton--Jacobi equation~\eqref{eq:ChaplyginHJ} then becomes
  \begin{equation}
    \label{eq:ChaplyginHJ-KnifeEdge}
    \frac{1}{2}\brackets{
      \frac{1}{m} \parentheses{ \pd{\bar{W}}{x} }^{2} + \frac{1}{J \cos^{2}\varphi} \parentheses{ \pd{\bar{W}}{\varphi} }^{2}
    }
    - m g x \sin \alpha = E.
  \end{equation}
  Assume that $\bar{W}: \bar{Q} \to \R$ takes the following form:
  \begin{equation*}
    \bar{W}(\varphi, x) =  \bar{W}_{\varphi}(\varphi) + \bar{W}_{x}(x).
  \end{equation*}
  Then, Eq.~\eqref{eq:ChaplyginHJ-KnifeEdge} becomes
  \begin{equation*}
    \frac{1}{2}\brackets{
      \frac{1}{m} \parentheses{ \od{\bar{W}_{x}}{x} }^{2}- (2 m g \sin \alpha)\,x
      + \frac{1}{J \cos^{2}\varphi} \parentheses{ \od{\bar{W}_{\varphi}}{\varphi} }^{2}
    }
    = E.
  \end{equation*}
  The first two terms in the brackets depend only on $x$, whereas the third only on $\varphi$, and thus
  \begin{equation*}
    \frac{1}{m} \parentheses{ \od{\bar{W}_{x}}{x} }^{2}- (2 m g \sin \alpha)\,x = 2E - \frac{(\gamma_{\varphi}^{0})^{2}}{J},
    \qquad
    \frac{1}{\cos^{2}\varphi} \parentheses{ \od{\bar{W}_{\varphi}}{\varphi} }^{2} = (\gamma_{\varphi}^{0})^{2},
  \end{equation*}
  with some positive constant $\gamma_{\varphi}^{0}$.
  Hence, assuming $\tod{\bar{W}_{x}}{x} \ge 0$, we have
  \begin{equation*}
    \od{\bar{W}_{x}}{x} = \sqrt{ m\parentheses{ 2E - \frac{(\gamma_{\varphi}^{0})^{2}}{J} } + (2 m^{2} g \sin\alpha)\,x },
    \qquad
    \od{\bar{W}_{\varphi}}{\varphi} = \gamma_{\varphi}^{0} \cos\varphi.
  \end{equation*}
  Then, Eq.~\eqref{eq:gamma-dbarW} gives
  \begin{equation*}
    \gamma(\varphi, x, y)
    = \gamma_{\varphi}^{0}\,d\varphi + \sqrt{ m\parentheses{ 2E - \frac{(\gamma_{\varphi}^{0})^{2}}{J} } + (2 m^{2} g \sin\alpha)\,x }\,( \cos\varphi\,dx + \sin\varphi\,dy),
  \end{equation*}
  which is the solution of the nonholonomic Hamilton--Jacobi equation~\eqref{eq:NHHJ} obtained in \citet[][Example~4.2]{OhBl2009}.
\end{example}

\section{Further Reduction and Hamiltonization}
\label{sec:FurtherReduction}
It often happens that there does not exist an $f$ that satisfies the necessary and sufficient condition in Proposition~\ref{prop:NSCondition-Hamiltonization} or the sufficient condition, Eq.~\eqref{eq:f}, and hence we cannot Hamiltonize the system based on the above theory.
However, we may reduce such systems further and then attempt to Hamiltonize the further-reduced system.

\subsection{Further Reduction of Chaplygin Systems}
We consider the following special case of the ``truncation'' of \citet[][Section~3.B]{HoGa2009}.
Recall the reduced Chaplygin system, Eq.~\eqref{eq:ReducedChaplyginSystem}, on $T^{*}\bar{Q}$, i.e., 
\begin{equation}
  \label{eq:ReducedChaplyginSystem2}
  i_{\bar{X}} \bar{\Omega}^{\rm nh} = d\bar{H}
\end{equation}
with the almost symplectic form
\begin{equation}
  \bar{\Omega}^{\rm nh} \defeq \bar{\Omega} - \Xi,
\end{equation}
and consider a free and proper Lie group action $K \times \bar{Q} \to \bar{Q}$, or $\Phi^{K}_{k}: \bar{Q} \to \bar{Q}$ with any $k \in K$, that satisfies the following conditions:
\begin{enumerate}[\bf I.]
\item The Hamiltonian $\bar{H}$ is $K$-invariant, i.e., $\bar{H} \circ T^{*}\Phi^{K}_{k} = \bar{H}$ for any $k \in K$, where $T^{*}\Phi^{K}$ is the cotangent lift of $\Phi^{K}$.
  \label{asmptn:K-invariance}
\item For any element $\eta$ in the Lie algebra $\mathfrak{k}$ of $K$, the infinitesimal generator $\eta_{T^{*}\bar{Q}}$ satisfies
  \begin{equation}
    \label{eq:i_eta-Xi}
    i_{\eta_{T^{*}\bar{Q}}} \Xi = 0.
  \end{equation}
  \label{asmptn:i_eta-Xi}
\end{enumerate}

Now, let ${\bf J}_{K}: T^{*}\bar{Q} \to \mathfrak{k}^{*}$ be the equivariant momentum map for the cotangent lift of the $K$-action $\Phi^{K}$, i.e., for any $\alpha_{\bar{q}} \in T^{*}\bar{Q}$ and $\eta \in \mathfrak{k}$,
\begin{equation}
  \label{eq:J_K-def}
  \ip{ {\bf J}_{K}(\alpha_{\bar{q}}) }{ \eta } = \ip{ \alpha_{\bar{q}} }{ \eta_{\bar{Q}} }.
\end{equation}
Also define $J_{K}^{\eta}: T^{*}\bar{Q} \to \R$ by $J_{K}^{\eta}(\alpha_{\bar{q}}) \defeq \ip{ {\bf J}_{K}(\alpha_{\bar{q}}) }{ \eta }$ for each $\eta \in \mathfrak{k}$.
Then, we have 
\begin{equation*}
  i_{\eta_{T^{*}\bar{Q}}} \bar{\Omega} = dJ_{K}^{\eta}.
\end{equation*}
Notice that Condition~\ref{asmptn:i_eta-Xi} implies
\begin{equation*}
  i_{\eta_{T^{*}\bar{Q}}} \bar{\Omega}^{\rm nh} = i_{\eta_{T^{*}\bar{Q}}} \bar{\Omega},
\end{equation*}
and thus
\begin{equation}
  \label{eq:J_K-Omega-Omega_nh}
  i_{\eta_{T^{*}\bar{Q}}} \bar{\Omega}^{\rm nh} = i_{\eta_{T^{*}\bar{Q}}} \bar{\Omega} = dJ_{K}^{\eta}.
\end{equation}
In other words, ${\bf J}_{K}$ is a momentum map with respect to {\em both} the standard symplectic form $\bar{\Omega}$ and the almost symplectic form $\bar{\Omega}^{\rm nh}$.
We also have the following:
\begin{proposition}
  \label{prop:J_K-conservation}
  Under Conditions~\ref{asmptn:K-invariance} and \ref{asmptn:i_eta-Xi} stated above, the momentum map ${\bf J}_{K}: T^{*}\bar{Q} \to \mathfrak{k}^{*}$ is conserved along the flow of the vector field $\bar{X}$ of the reduced Chaplygin system, Eq.~\eqref{eq:ReducedChaplyginSystem2}.
\end{proposition}

\begin{proof}
  Follows easily from the following calculation:
  \begin{align*}
    \bar{X}[J_{K}^{\eta}] &= i_{\bar{X}} dJ_{K}^{\eta}
    \\
    &= i_{\bar{X}} i_{\eta_{T^{*}\bar{Q}}} \bar{\Omega}^{\rm nh}
    \\
    &= -i_{\eta_{T^{*}\bar{Q}}} i_{\bar{X}} \bar{\Omega}^{\rm nh}
    \\
    &= -i_{\eta_{T^{*}\bar{Q}}} d\bar{H}
    \\
    &= -\eta_{T^{*}\bar{Q}} [\bar{H}]
    \\
    &= 0,
  \end{align*}
  where we used Eq.~\eqref{eq:J_K-Omega-Omega_nh} in the second line, and Condition~\ref{asmptn:K-invariance} in the last line.
\end{proof}

Also, let $K_{\mu}$ be the coadjoint isotropy group of $\mu$, i.e., $K_{\mu} \defeq \setdef{k \in K}{\Ad_{k}^{*} \mu = \mu}$, and assume
\begin{enumerate}[\bf I.]
  \setcounter{enumi}{2}
\item $\mu \in \mathfrak{k}^{*}$ is a regular value of ${\bf J}_{K}$, and $K_{\mu}$ acts freely and properly on ${\bf J}_{K}^{-1}(\mu)$.
  \label{asmptn:K-free_and_proper}
\end{enumerate}
Since ${\bf J}_{K}$ is a momentum map with respect to the almost symplectic form $\bar{\Omega}^{\rm nh}$, the two-form $\bar{\Omega}^{\rm nh}$ itself works as a ``truncated form'' (see \citet[][Section~3.B and Theorem~3.3]{HoGa2009}) in this special case: Performing the almost symplectic reduction of \citet{Pl2004}, we may drop the dynamics to ${\bf J}_{K}^{-1}(\mu)/K_{\mu}$ as follows:

\begin{proposition}[Further Reduction of Chaplygin Systems]
  \label{prop:FurtherReduction}
  Under Conditions~\ref{asmptn:K-invariance}--\ref{asmptn:K-free_and_proper}, we have the following:
  \begin{enumerate}[(i)]
  \item There exists an almost symplectic form $\bar{\Omega}^{\rm nh}_{\mu}$ on ${\bf J}_{K}^{-1}(\mu)/K_{\mu}$ uniquely
    characterized by
    \begin{equation}
      \label{eq:barOmegaNH-barOmega_muNH}
      \pi_{\mu}^{*} \bar{\Omega}^{\rm nh}_{\mu} = i_{\mu}^{*} \bar{\Omega}^{\rm nh},
    \end{equation}
    where $i_{\mu}: {\bf J}_{K}^{-1}(\mu) \hookrightarrow T^{*}\bar{Q}$ and $\pi_{\mu}: {\bf J}_{K}^{-1}(\mu) \to {\bf J}_{K}^{-1}(\mu)/K_{\mu}$.
    \label{prop:FurtherReduction-i}
    \medskip
  \item The reduced Chaplygin system, Eq.~\eqref{eq:ReducedChaplyginSystem2}, is further reduced to the following system:
    \begin{equation}
      \label{eq:FurtherReducedChaplyginSystem}
      i_{\bar{X}_{\mu}} \bar{\Omega}^{\rm nh}_{\mu} = d\bar{H}_{\mu},
    \end{equation}
    where $\bar{X}$ and $\bar{X}_{\mu}$ are $\pi_{\mu}$-related, i.e.,
    \begin{equation}
      \label{eq:barX-barX_mu}
      T\pi_{\mu} \circ \bar{X} = \bar{X}_{\mu} \circ \pi_{\mu},
    \end{equation}
    and $\bar{H}_{\mu}: {\bf J}_{K}^{-1}(\mu)/K_{\mu} \to \R$ is defined by
    \begin{equation}
      \label{eq:barH-barH_mu}
      \bar{H}_{\mu} \circ \pi_{\mu} = \bar{H} \circ i_{\mu}.
    \end{equation}
    \label{prop:FurtherReduction-ii}
  \item The almost symplectic form $ \bar{\Omega}^{\rm nh}_{\mu}$ is written as
    \begin{equation*}
      \bar{\Omega}^{\rm nh}_{\mu} = \bar{\Omega}_{\mu} - \Xi_{\mu},
    \end{equation*}
    where $\Xi_{\mu}$ is uniquely characterized by
    \begin{equation}
      \label{eq:Xi-Xi_mu}
      \pi_{\mu}^{*} \Xi_{\mu} = i_{\mu}^{*} \Xi.
    \end{equation}
    \label{prop:FurtherReduction-iii}
  \end{enumerate}
\end{proposition}

\begin{proof}
  \eqref{prop:FurtherReduction-i} and \eqref{prop:FurtherReduction-ii} follow directly from \citet[][Theorem~2.1]{Pl2004}.
  \eqref{prop:FurtherReduction-iii}~Since ${\bf J}_{K}$ is an equivariant momentum map with respect to the {\em canonical} symplectic form $\bar{\Omega}$, the symplectic reduction of \citet{MaWe1974} applies here as well (not to the reduction of the dynamics but to the reduction of the symplectic structure).
  Hence there exists a unique (strictly) symplectic form $\bar{\Omega}_{\mu}$ on ${\bf J}_{K}^{-1}(\mu)/K_{\mu}$ such that
  \begin{equation}
    \label{eq:barOmega-barOmega_mu}
    \pi_{\mu}^{*} \bar{\Omega}_{\mu} = i_{\mu}^{*} \bar{\Omega}.
  \end{equation}
  Combining this with Eq.~\eqref{eq:barOmegaNH-barOmega_muNH}, we have
  \begin{equation*}
    \pi_{\mu}^{*} (\bar{\Omega}_{\mu} - \bar{\Omega}^{\rm nh}_{\mu})
    = i_{\mu}^{*} (\bar{\Omega} - \bar{\Omega}^{\rm nh})
    = i_{\mu}^{*} \Xi.
  \end{equation*}
  Since $\pi_{\mu}$ is a surjective submersion, the pull-back $\pi_{\mu}^{*}$ is injective, and thus the uniqueness follows.
\end{proof}

Furthermore, under certain assumptions, we may employ a result from the theory of cotangent bundle reduction (see, e.g., \citet[Section~2.2]{MaMiOrPeRa2007}) to make our result more explicit.
To that end, we first define a mechanical connection on the principal bundle
\begin{equation*}
  \bar{\pi}: \bar{Q} \to \bar{Q}/K \eqdef \tilde{Q}
\end{equation*}
as follows: For each $\bar{q} \in \bar{Q}$, let $\mathbb{I}(\bar{q}): \mathfrak{k} \to \mathfrak{k}^{*}$ be the locked inertia tensor defined by
\begin{equation*}
  \ip{ \mathbb{I}(\bar{q}) \eta }{ \zeta } = \bar{g}_{\bar{q}}\parentheses{ \eta_{\bar{Q}}(\bar{q}), \zeta_{\bar{Q}}(\bar{q}) },
\end{equation*}
where $\bar{g}$ is the kinetic energy metric defined in Eq.~\eqref{eq:barg}, and $\eta$ and $\zeta$ are arbitrary elements in $\mathfrak{k}$.
Then, the mechanical connection $\mathcal{A}_{K}: T\bar{Q} \to \mathfrak{k}$ is defined by
\begin{equation}
  \label{eq:mathcalA_K}
  \mathcal{A}_{K}(v_{\bar{q}}) \defeq \mathbb{I}(\bar{q})^{-1} \circ {\bf J}_{K} \parentheses{ \F\bar{L}(v_{\bar{q}}) }.
\end{equation}
We will also need the ``$\mu$-component'' of $\mathcal{A}_{K}$, i.e., the one-form $\alpha_{\mu}$ on $\bar{Q}$ defined by $\alpha_{\mu}(\bar{q}) \defeq \mathcal{A}_{K}(\bar{q})^{*}\mu$, or equivalently, 
\begin{equation}
  \label{eq:alpha_mu}
  \ip{\alpha_{\mu}(\bar{q})}{v_{\bar{q}}} = \ip{\mu}{\mathcal{A}_{K}(v_{\bar{q}})}.
\end{equation}
Let us introduce the two-form $\beta_{\mu}$ on $\tilde{Q}$ defined by
\begin{equation}
  \label{eq:beta_mu}
  \bar{\pi}^{*} \beta_{\mu} = d\alpha_{\mu},
\end{equation}
and also the two-form $B^{K}_{\mu}$ on $T^{*}\tilde{Q}$ defined by
\begin{equation}
  \label{eq:B^K_mu}
  B^{K}_{\mu} \defeq \pi_{\tilde{Q}}^{*} \beta_{\mu},
\end{equation}
where $\pi_{\tilde{Q}}: T^{*}\tilde{Q} \to \tilde{Q}$ is the cotangent bundle projection.

Now, we assume the following:
\begin{enumerate}[\bf I.]
  \setcounter{enumi}{3}
\item $K_{\mu} = K$, which is always the case when $K$ is Abelian;
  \label{asmptn:K_mu=K}
\item $\alpha_{\mu}$ is $K$-invariant and takes values in ${\bf J}_{K}^{-1}(\mu)$.
  \label{asmptn:alpha_mu}
\end{enumerate}
With these additional assumptions, we have the following important special case of Proposition~\ref{prop:FurtherReduction}:
\begin{proposition}[Further Reduction of Chaplygin Systems---Special Case]
  \label{prop:FurtherReduction-K_mu=K}
  If, in addition, Conditions~\ref{asmptn:K_mu=K} and \ref{asmptn:alpha_mu} hold, then we may extend the results of Proposition~\ref{prop:FurtherReduction} so that the dynamics after the second reduction is described on the cotangent bundle $T^{*}(\bar{Q}/K) = T^{*}\tilde{Q}$ as follows:
  \begin{enumerate}[(i)]
  \item The reduced space ${\bf J}_{K}^{-1}(\mu)/K$ is symplectically diffeomorphic to $T^{*}\tilde{Q}$ with the symplectic structure $\tilde{\Omega} - B^{K}_{\mu}$, where $\tilde{\Omega}$ is the standard symplectic form on $T^{*}\tilde{Q}$.
    \label{prop:FurtherReduction-K_mu=K-i}
    \medskip
  \item Let $\varphi_{\mu}: {\bf J}_{K}^{-1}(\mu)/K \to T^{*}\tilde{Q}$ be the symplectomorphism that gives the correspondence in \eqref{prop:FurtherReduction-K_mu=K-i}.
    Then, the dynamics on ${\bf J}_{K}^{-1}(\mu)/K$ defined by Eq.~\eqref{eq:FurtherReducedChaplyginSystem} is equivalent to the one defined by
    \begin{equation}
      \label{eq:FurtherReducedChaplyginSystem-K_mu=K}
      i_{\tilde{X}_{\mu}} \tilde{\Omega}^{\rm nh}_{\mu} = d\tilde{H}_{\mu},
    \end{equation}
    where
    \begin{equation}
      \label{eq:tildeOmega^nh_mu}
      \tilde{\Omega}^{\rm nh}_{\mu} \defeq \tilde{\Omega} - B^{K}_{\mu} - \tilde{\Xi}_{\mu}
    \end{equation}
    with $\tilde{X}_{\mu} \defeq (\varphi_{\mu}^{-1})^{*} \bar{X}_{\mu}$, $\tilde{\Xi}_{\mu} \defeq (\varphi_{\mu}^{-1})^{*} \Xi_{\mu}$, and $\tilde{H}_{\mu} \defeq \bar{H}_{\mu} \circ \varphi_{\mu}^{-1}$.
    \label{prop:FurtherReduction-K_mu=K-ii}
  \end{enumerate}
\end{proposition}

\begin{proof}
  \eqref{prop:FurtherReduction-K_mu=K-i}~See \citet[Theorem~2.2.3 on p.~64]{MaMiOrPeRa2007}.
  The construction of the map $\varphi_{\mu}$ is summarized in Appendix~\ref{sec:varphi_mu}.
  The diagrams below summarize the spaces and almost symplectic forms involved in the procedure of the reduction.
  \begin{equation*}
    \xymatrix@!0@R=0.65in@C=1in{
      T^{*}\bar{Q} & 
      \\
      {\bf J}_{K}^{-1}(\mu) \ar[u]^{i_{\mu}} \ar[d]_{\pi_\mu} &
      \\
      {\bf J}_{K}^{-1}(\mu)/K \ar[r]^{\quad\ \varphi_\mu} & T^{*}\tilde{Q}
    }
    \qquad
    \xymatrix@!0@R=0.65in@C=1in{
      \bar{\Omega}^{\rm nh} \ar@{|->}[d]_{i_{\mu}^{*}} & 
      \\
      i_{\mu}^{*}\bar{\Omega}^{\rm nh} = \pi_{\mu}^{*} \bar{\Omega}^{\rm nh}_{\mu} &
      \\
      \bar{\Omega}^{\rm nh}_{\mu} \ar@{|->}[u]^{\pi_\mu^{*}} \ar@{|->}[r]^{(\varphi_\mu^{-1})^{*}} & \tilde{\Omega}^{\rm nh}_{\mu}
    }
    \end{equation*}
  \eqref{prop:FurtherReduction-K_mu=K-ii}~Apply $(\varphi_{\mu}^{-1})^{*}$ to both sides of Eq.~\eqref{eq:FurtherReducedChaplyginSystem} and use the fact from \eqref{prop:FurtherReduction-K_mu=K-i} that $(\varphi_{\mu}^{-1})^{*} \bar{\Omega}_{\mu} = \tilde{\Omega} - B^{K}_{\mu}$.
\end{proof}

\subsection{Hamiltonization after Second Reduction}
Now, we follow a similar argument as in Section~\ref{chsb} to discuss the Hamiltonizability of the system defined by Eq.~\eqref{eq:FurtherReducedChaplyginSystem-K_mu=K}: Let $f_{\mu}: T^{*}\tilde{Q} \to \R$ be a smooth nowhere-vanishing function that is constant on each fiber, and define the map $\tilde{\Psi}_{\!f_{\mu}}: T^{*}\tilde{Q} \to T^{*}\tilde{Q}$ by
\begin{equation*}
  \tilde{\Psi}_{\!f_{\mu}}: \alpha \mapsto f_{\mu}\,\alpha.
\end{equation*}
Define the vector field $\tilde{X}_{\rm C}^{\mu}$ analogously to Eq.~\eqref{eq:barX_C} so that
\begin{equation}
  \label{eq:tildeX_C^mu}
  \tilde{X}_{\rm C}^{\mu} = \tilde{\Psi}_{1/f_{\mu}}^{*} (\tilde{X}_{\mu}/f_{\mu}),
\end{equation}
and hence $\tilde{X}_{\mu}/f_{\mu}$ and $\tilde{X}_{\rm C}^{\mu}$ are $\tilde{\Psi}_{\!f_{\mu}}$-related:
\begin{equation}
  \label{eq:tildePsi_f_mu-related}
  T\tilde{\Psi}_{\!f_{\mu}} \circ (\tilde{X}_{\mu}/f_{\mu}) = \tilde{X}_{\rm C}^{\mu} \circ \tilde{\Psi}_{\!f_{\mu}}.
\end{equation}
Following the same arguments as in the proofs of Proposition~\ref{prop:NSCondition-Hamiltonization} and Theorem~\ref{thm:SufficientCondition}, we obtain similar results for the further-reduced system, Eq.~\eqref{eq:FurtherReducedChaplyginSystem-K_mu=K}.
\begin{proposition}[Necessary and Sufficient Condition for Hamiltonization after Second Reduction]
  \label{prop:NSCondition-Hamiltonization2}
  The vector field $\tilde{X}_{\rm C}^{\mu} \in \mathfrak{X}(T^{*}\tilde{Q})$ satisfies Hamilton's equations
  \begin{equation*}
    i_{\tilde{X}_{\rm C}^{\mu}} \tilde{\Omega} = d\tilde{H}_{\rm C}^{\mu}
  \end{equation*}
  with the Chaplygin Hamiltonian
  \begin{equation}
    \label{eq:ChaplyginHamiltonian2}
    \tilde{H}_{\rm C}^{\mu} \defeq \tilde{H}_{\mu} \circ \tilde{\Psi}_{1/f_{\mu}}
  \end{equation}
  if and only if the one-form $i_{\tilde{X}_{\rm C}^{\mu}}\brackets{ df_{\mu} \wedge \tilde{\Theta} - f_{\mu}^{2} \parentheses{ B^{K}_{\mu} + \tilde{\Psi}_{1/f_{\mu}}^{*} \tilde{\Xi}_{\mu} } }$ vanishes, where $\tilde{\Theta}$ is the symplectic one-form on $T^{*}\tilde{Q}$.
\end{proposition}

\begin{proof}
  The result follows from essentially the same calculations as in the proof of Lemma~\ref{lem:barX_C-eq}.
  The only difference is the treatment of the curvature term $B^{K}_{\mu}$, which is not present in Lemma~\ref{lem:barX_C-eq}.
  Specifically, we need to calculate $\tilde{\Psi}_{1/f_{\mu}}^{*} B^{K}_{\mu}$: From the definition of $B^{K}_{\mu}$, Eq.~\eqref{eq:B^K_mu}, we have
  \begin{align*}
    \tilde{\Psi}_{1/f_{\mu}}^{*} B^{K}_{\mu}
    &= \tilde{\Psi}_{1/f_{\mu}}^{*} \pi_{\tilde{Q}}^{*} \beta_{\mu} 
    \\
    &=  \parentheses{ \pi_{\tilde{Q}} \circ \tilde{\Psi}_{1/f_{\mu}} }^{*} \beta_{\mu} 
    \\
    &=  \pi_{\tilde{Q}}^{*} \beta_{\mu} 
    \\
    &= B^{K}_{\mu},
  \end{align*}
  where we used the fact that $\tilde{\Psi}_{1/f_{\mu}}$ is fiber-preserving, i.e., $\pi_{\tilde{Q}} \circ \tilde{\Psi}_{1/f_{\mu}} = \pi_{\tilde{Q}}$.
  Therefore, we obtain
  \begin{equation}
    \label{eq:tildeY_mu-eq}
    i_{\tilde{X}_{\rm C}^{\mu}} \braces{
      \tilde{\Omega}
      + \frac{1}{f_{\mu}} \brackets{
        df_{\mu} \wedge \tilde{\Theta}
        - f_{\mu}^{2} \parentheses{ B^{K}_{\mu} + \tilde{\Psi}_{1/f_{\mu}}^{*} \tilde{\Xi}_{\mu} }
      }
    }
    =
    d\tilde{H}_{\rm C}^{\mu},
  \end{equation}
  and thus the claim follows.
\end{proof}

\begin{remark}
  Since $\tilde{\Omega} - B^{K}_{\mu}$ is also a (non-standard) symplectic form as well, we may discuss Hamiltonization with respect to this symplectic  form.
  However, we prefer to work with the standard symplectic form $\tilde{\Omega}$ since the standard Hamilton--Jacobi theory directly applies to Hamiltonian systems defined with the standard symplectic form $\tilde{\Omega}$.
\end{remark}

\begin{theorem}[A Sufficient Condition for Hamiltonization after Second Reduction]
  \label{thm:SufficientCondition2}
  Suppose there exists a nowhere-vanishing fiber-wise constant function $f_{\mu}: T^{*}\tilde{Q} \to \R$ that satisfies the equation
  \begin{equation}
    \label{eq:f_mu}
    df_{\mu} \wedge \tilde{\Theta} = f_{\mu}^{2} \parentheses{ B^{K}_{\mu} + \tilde{\Psi}_{1/f_{\mu}}^{*} \tilde{\Xi}_{\mu} }.
  \end{equation}
  Then, the vector field $\tilde{X}_{\rm C}^{\mu} \in \mathfrak{X}(T^{*}\tilde{Q})$ (see Eq.~\eqref{eq:tildeX_C^mu}) satisfies the following Hamilton's equations:
  \begin{equation}
    \label{eq:tildeY_mu-HamiltonEq}
    i_{\tilde{X}_{\rm C}^{\mu}} \tilde{\Omega} = d\tilde{H}_{\rm C}^{\mu},
  \end{equation}
  and, as a result, the further-reduced nonholonomic dynamics, Eq.~\eqref{eq:FurtherReducedChaplyginSystem-K_mu=K}, has the invariant measure $f_{\mu}^{\tilde{n}-1} \tilde{\Lambda}$, where $\tilde{n} \defeq \dim \tilde{Q}$ and
  \begin{equation*}
    \tilde{\Lambda} \defeq \frac{(-1)^{\tilde{n}(\tilde{n}-1)/2}}{\tilde{n}!} \underbrace{\tilde{\Omega} \wedge \dots \wedge \tilde{\Omega}}_{\tilde{n}}.
  \end{equation*}
\end{theorem}

\begin{proof}
  Follows immediately from Eq.~\eqref{eq:tildeY_mu-eq} and Corollary~\ref{cor:FeJo2004}.
\end{proof}

\section{Nonholonomic H--J Theory via Hamiltonization after Second Reduction}
\label{sec:NHHJafter2ndReduction}
\subsection{Relationship between the Chaplygin H--J and Nonholonomic H--J Equations after Second Reduction}\label{hjnhsec2}
If the system is Hamiltonized in the sense of Theorem~\ref{thm:SufficientCondition2}, then we have the Chaplygin Hamilton--Jacobi equation 
\begin{equation}
  \label{eq:ChaplyginHJ2}
  \tilde{H}^{\mu}_{\rm C} \circ d\tilde{W}^{\mu} = E
\end{equation}
corresponding to Hamilton's equation \eqref{eq:tildeY_mu-HamiltonEq}.
One then wonders if there is any relationship between $\tilde{W}^{\mu}$ and $\gamma$ that is similar to the one obtained in Theorem~\ref{thm:ChaplyginHJ-NHHJ}.

A natural starting point towards the answer to this question is, again, to look into the relationship between the Chaplygin Hamiltonian $\tilde{H}^{\mu}_{\rm C}$ and the original Hamiltonian $H$; then we obtain the relationship between the two solutions $\tilde{W}^{\mu}$ and $\gamma$ by exploiting the geometry involved in the process of reduction and Hamiltonization.
The diagram below combines the following things together: the first and second reductions of Chaplygin systems; the relationship between the two Hamiltonians $\tilde{H}^{\mu}_{\rm C}$ and $H$; the shift map ${\rm shift}_{\mu}: {\bf J}_{K}^{-1}(\mu) \to {\bf J}_{K}^{-1}(0)$ (see Appendix~\ref{sec:varphi_mu}); also the horizontal lift $\hl^{\bar{\mathcal{M}}}: T^{*}Q \to \bar{\mathcal{M}} \defeq {\bf J}_{K}^{-1}(0)$, which is defined in a similar way as $\hl^{\mathcal{M}}$ (see Eq.~\eqref{eq:hl^M}) using the connection $\mathcal{A}_{K}$ (see Eq.~\eqref{eq:mathcalA_K}) as follows:
Let us define the horizontal space 
\begin{equation*}
  \bar{\mathcal{D}} \defeq \ker\mathcal{A}_{K} \subset T\bar{Q}.
\end{equation*}
Then, the connection $\mathcal{A}_{K}$ induces the horizontal lift $\hl^{\bar{\mathcal{D}}}: T\tilde{Q} \to \bar{\mathcal{D}}$ defined by $\hl^{\bar{\mathcal{D}}} \defeq (T\bar{\pi}|_{\bar{\mathcal{D}}})^{-1}$.
Let us also define
\begin{equation*}
  \bar{\mathcal{M}} \defeq {\bf J}_{K}^{-1}(0).
\end{equation*}
Then, it is straightforward to see that $\bar{\mathcal{M}} = \mathbb{F}\bar{L}(\bar{\mathcal{D}})$.
Now, we define the horizontal lift $\hl^{\bar{\mathcal{M}}}: T^{*}\tilde{Q} \to \bar{\mathcal{M}}$ as follows:
\begin{equation}
  \label{eq:hl^barM}
  \hl^{\bar{\mathcal{M}}}_{\bar{q}}
  \defeq \F\bar{L}_{\bar{q}} \circ \hl^{\bar{\mathcal{D}}}_{\bar{q}} \circ (\F\tilde{L})^{-1}_{\tilde{q}},
\end{equation}
where $\tilde{L}: T\tilde{Q} \to \R$ is defined by $\tilde{L} \defeq \bar{L} \circ \hl^{\bar{\mathcal{D}}}$.
\begin{equation}
  \label{dia:gamma-dtildeW_mu}
  \vcenter{
    \xymatrix@!0@R=0.65in@C=0.85in{
      & & \R & & &
      \\
      & & & & &
      \\
      \mathcal{M} \ar[rruu]^{H\!\!} & T^{*}\bar{Q} \ar[l]^{\hl^{\mathcal{M}}} & {\bf J}_{K}^{-1}(\mu) \ar[l]_{i_{\mu}} \ar[d]_{\pi_\mu} & {\bf J}_{K}^{-1}(0) \ar[l]_{\ {\rm shift}_{\mu}^{-1}} & &
      \\
      & & {\bf J}_{K}^{-1}(\mu)/K & & T^{*}\tilde{Q} \ar[ll]^{\quad\ \varphi_\mu^{-1}} \ar[lu]_{\hl^{\bar{\mathcal{M}}}} & T^{*}\tilde{Q} \ar[l]^{\ \,\tilde{\Psi}_{1/f_\mu}} \ar[llluuu]_{\!\!\!\tilde{H}_{\rm C}^{\mu}}
      \\
      Q \ar[r]_{\pi} \ar@{-->}[uu]^{\gamma} & \bar{Q} \ar[rrrr]_{\bar{\pi}} \ar@{-->}[uu]^{\bar{\gamma}_\mu} & & & & \tilde{Q} \ar[u]_{d\tilde{W}^{\mu}}
    }
  }
\end{equation}
That the map $\hl^{\bar{\mathcal{M}}}$ fits into the diagram is shown in Appendix~\ref{sec:hl^barM} (see also Appendix~\ref{sec:varphi_mu}).
The diagram also shows the map $d\tilde{W}^{\mu}: \tilde{Q} \to T^{*}\tilde{Q}$ with $\tilde{W}^{\mu}$ being a solution of the Chaplygin Hamilton--Jacobi equation~\eqref{eq:ChaplyginHJ2}; this leads us to the following result that is similar to Theorem~\ref{thm:ChaplyginHJ-NHHJ}:

\begin{theorem}
  \label{thm:ChaplyginHJ2-NHHJ}
  Suppose that there exists a nowhere-vanishing fiber-wise constant function $f_{\mu}: T^{*}\tilde{Q} \to \R$ that satisfies Eq.~\eqref{eq:f_mu}, and hence by Theorem~\ref{thm:SufficientCondition2}, we have Hamilton's equations \eqref{eq:tildeY_mu-HamiltonEq} for the vector field $\tilde{X}_{\rm C}^{\mu}$.
  Let $\tilde{W}^{\mu}: \tilde{Q} \to \R$ be a solution of the Chaplygin Hamilton--Jacobi equation~\eqref{eq:ChaplyginHJ2}, and define $\gamma: Q \to \mathcal{M}$ by
  \begin{equation}
    \label{eq:gamma-bargamma_mu}
    \gamma(q) \defeq
    \hl^{\mathcal{M}}_{q} \circ \bar{\gamma}_{\mu} \circ \pi(q)
  \end{equation}
  with $\bar{\gamma}_{\mu}: \bar{Q} \to T^{*}\bar{Q}$ defined by\footnote{See Appendix~\ref{sec:varphi_mu} for the relationship between $i_{\mu}$, $i_{0}$, and ${\rm shift}_{\mu}$: We have $i_{\mu} \circ {\rm shift}_{\mu}^{-1}(p_{\bar{q}}) = i_{0}(p_{\bar{q}}) + \alpha_{\mu}(\bar{q})$ for any $p_{\bar{q}} \in {\bf J}_{K}^{-1}(0)$.}
  \begin{align}
    \bar{\gamma}_{\mu}(\bar{q})
    &\defeq i_{\mu} \circ {\rm shift}_{\mu}^{-1} \circ \hl^{\bar{\mathcal{M}}}_{\bar{q}} \circ \tilde{\Psi}_{1/f_{\mu}} \circ d\tilde{W}^{\mu} \circ \bar{\pi}(\bar{q})
    \nonumber\\
    &= i_{0} \circ \hl^{\bar{\mathcal{M}}}_{\bar{q}} \parentheses{ \frac{1}{f_{\mu}} d\tilde{W}^{\mu}(\tilde{q}) }
    + \alpha_{\mu}(\bar{q}),
    \label{eq:bargamma_mu-dtildeW_mu}
  \end{align}
  where $\bar{q} \defeq \pi(q)$ and $\tilde{q} \defeq \bar{\pi}(\bar{q})$.
  Then $\gamma$ satisfies the nonholonomic Hamilton--Jacobi equation~\eqref{eq:NHHJ} as well as the condition Eq.~\eqref{eq:dgamma}.
\end{theorem}

\begin{proof}
  That the one-form $\gamma$ defined by Eqs.~\eqref{eq:gamma-bargamma_mu} and \eqref{eq:bargamma_mu-dtildeW_mu} satisfies the nonholonomic Hamilton--Jacobi equation~\eqref{eq:NHHJ} follows from the diagram \eqref{dia:gamma-dtildeW_mu}.
  Showing that it also satisfies the condition Eq.~\eqref{eq:dgamma} requires tedious calculations:
  Following a similar calculation to that of $d\gamma(Y^{\rm h}, Z^{\rm h})$ in the proof of Theorem~\ref{thm:ChaplyginHJ-NHHJ}, Eq.~\eqref{eq:gamma-bargamma_mu} gives
  \begin{equation}
    \label{eq:dgamma2}
    d\gamma(Y^{\rm h},Z^{\rm h}) = d\bar{\gamma}_{\mu}(Y, Z) + \bar{\gamma}_{\mu}^{*} \Xi(Y, Z)
  \end{equation}
  for arbitrary horizontal vector fields $Y^{\rm h}, Z^{\rm h} \in \mathfrak{X}(Q)$, where $Y \defeq T\pi(Y^{\rm h})$ and similarly for $Z$.

  Let us calculate the first term in Eq.~\eqref{eq:dgamma2}:
  Writing
  \begin{equation*}
    \bar{\gamma}_{0} \defeq i_{0} \circ \hl^{\bar{\mathcal{M}}}_{\bar{q}} \parentheses{ \frac{1}{f_{\mu}} d\tilde{W}^{\mu} },
  \end{equation*}
  we have $\bar{\gamma}_{\mu} = \bar{\gamma}_{0} + \alpha_{\mu}$ and thus
  \begin{equation*}
    d\bar{\gamma}_{\mu}(Y, Z)
    = d\bar{\gamma}_{0}(Y, Z)
    + d\alpha_{\mu}(Y, Z).
  \end{equation*}
  Calculation of $d\bar{\gamma}_{0}(Y, Z)$ is somewhat similar to that of $d\gamma(Y^{\rm h}, Z^{\rm h})$ in the proof of Theorem~\ref{thm:ChaplyginHJ-NHHJ}, but there is one difference: $Y$ and $Z$ are not horizontal here.
  Specifically, we have
  \begin{align*}
    d\bar{\gamma}_{0}(Y, Z)
    &= Y[\bar{\gamma}_{0}(Z)] - Z[\bar{\gamma}_{0}(Y)] - \bar{\gamma}_{0}([Y,Z])
    \\
    &= \frac{1}{f_{\mu}}\, d\tilde{W}_{\mu}([\tilde{Y},\tilde{Z}])
    - \frac{1}{f_{\mu}^{2}}\, df_{\mu} \wedge d\tilde{W}_{\mu}(\tilde{Y}, \tilde{Z})
    - \bar{\gamma}_{0}([Y,Z]),
  \end{align*}
  where we defined $\tilde{Y} \defeq T\bar{\pi}(Y)$ and similarly for $\tilde{Z}$.
  To calculate $\bar{\gamma}_{0}([Y,Z])$, we decompose $[Y,Z]$ into the horizontal and vertical parts:
  \begin{equation*}
    [Y,Z] = \hl^{\bar{\mathcal{D}}}([\tilde{Y}, \tilde{Z}]) + \parentheses{ \mathcal{A}_{K}([Y,Z]) }_{\bar{Q}},
  \end{equation*}
  where we note that $T\bar{\pi}([Y,Z]) = [\tilde{Y}, \tilde{Z}]$, since $Y$ and $Z$ are $\bar{\pi}$-related to $\tilde{Y}$ and $\tilde{Z}$, respectively.
  As a result, we have
  \begin{align*}
    \bar{\gamma}_{0}([Y,Z])(\bar{q})
    &= \frac{1}{f_{\mu}(\tilde{q})}\, d\tilde{W}_{\mu}([\tilde{Y},\tilde{Z}])(\tilde{q})
    + \frac{1}{f_{\mu}(\tilde{q})} \ip{ {\bf J}_{K} \circ \hl^{\mathcal{\bar{M}}}_{\bar{q}} \parentheses{ d\tilde{W}_{\mu}(\tilde{q})} }{ \mathcal{A}_{K}([Y,Z])(\bar{q}) }
    \\
    &= \frac{1}{f_{\mu}(\tilde{q})}\, d\tilde{W}_{\mu}([\tilde{Y},\tilde{Z}])(\tilde{q}),
  \end{align*}
  where the second term vanishes because $\hl^{\bar{\mathcal{M}}}$ takes values in $\bar{\mathcal{M}} \defeq {\bf J}_{K}^{-1}(0)$.
  Next let us calculate $d\alpha_{\mu}(Y, Z)$: Using Eq.~\eqref{eq:beta_mu}, the relation $\pi_{\tilde{Q}} \circ d\tilde{W}^{\mu} = \id_{\tilde{Q}}$, and Eq.~\eqref{eq:B^K_mu}, we obtain
  \begin{align*}
    d\alpha_{\mu}(Y, Z)
    &= \bar{\pi}^{*} \beta_{\mu}(Y, Z)
    \\
    &= \bar{\pi}^{*} \circ (\pi_{\tilde{Q}} \circ d\tilde{W}^{\mu})^{*} \beta_{\mu}(Y, Z).
    \\
    &= \bar{\pi}^{*} \circ (d\tilde{W}^{\mu})^{*} \circ \pi_{\tilde{Q}}^{*} \beta_{\mu}(Y, Z).
    \\
    &= \bar{\pi}^{*} \circ (d\tilde{W}^{\mu})^{*} B^{K}_{\mu}(Y, Z).
    \\
    &= (d\tilde{W}^{\mu})^{*} B^{K}_{\mu}(\tilde{Y}, \tilde{Z}).
  \end{align*}
  Therefore, the first term on the right-hand side of Eq.~\eqref{eq:dgamma2} becomes
  \begin{equation*}
    d\bar{\gamma}_{\mu}(Y, Z)
    = -\frac{1}{f_{\mu}^{2}} (d\tilde{W}^{\mu})^{*} \parentheses{ df_{\mu} \wedge \tilde{\Theta} - f_{\mu}^{2}\, B^{K}_{\mu} }(\tilde{Y}, \tilde{Z}),
  \end{equation*}
  since $(d\tilde{W}_{\mu})^{*}f_{\mu} = f_{\mu}$ and also $(d\tilde{W}_{\mu})^{*} \tilde{\Theta} = d\tilde{W}_{\mu}$.

  Now, let us evaluate the second term on the right-hand side of Eq.~\eqref{eq:dgamma2}: Substitution of Eq.~\eqref{eq:bargamma_mu-dtildeW_mu} gives
  \begin{align*}
    \bar{\gamma}_{\mu}^{*} \Xi
    &= \parentheses{ {\rm shift}_{\mu}^{-1} \circ \hl^{\bar{\mathcal{M}}}_{\bar{q}} \circ \tilde{\Psi}_{1/f_{\mu}} \circ d\tilde{W}^{\mu} \circ \bar{\pi} }^{*} \circ i_{\mu}^{*} \Xi
    \\
    &= \parentheses{ {\rm shift}_{\mu}^{-1} \circ \hl^{\bar{\mathcal{M}}}_{\bar{q}} \circ \tilde{\Psi}_{1/f_{\mu}} \circ d\tilde{W}^{\mu} \circ \bar{\pi} }^{*} \circ \pi_{\mu}^{*} \Xi_{\mu}
    \\
    &= \parentheses{ \pi_{\mu} \circ {\rm shift}_{\mu}^{-1} \circ \hl^{\bar{\mathcal{M}}}_{\bar{q}} \circ \tilde{\Psi}_{1/f_{\mu}} \circ d\tilde{W}^{\mu} \circ \bar{\pi} }^{*} \Xi_{\mu}
    \\
    &= \parentheses{ \tilde{\Psi}_{1/f_{\mu}} \circ d\tilde{W}^{\mu} \circ \bar{\pi} }^{*} \circ \parentheses{ \pi_{\mu} \circ {\rm shift}_{\mu}^{-1} \circ \hl^{\bar{\mathcal{M}}}_{\bar{q}} }^{*} \Xi_{\mu}
    \\
    &= \parentheses{ \tilde{\Psi}_{1/f_{\mu}} \circ d\tilde{W}^{\mu} \circ \bar{\pi} }^{*} \circ \parentheses{ \varphi_{\mu}^{-1} }^{*} \Xi_{\mu}
    \\
    &= \parentheses{ \tilde{\Psi}_{1/f_{\mu}} \circ d\tilde{W}^{\mu} \circ \bar{\pi} }^{*} \tilde{\Xi}_{\mu}
    \\
    &= \bar{\pi}^{*} \circ (d\tilde{W}^{\mu})^{*} \circ \tilde{\Psi}_{1/f_{\mu}}^{*} \tilde{\Xi}_{\mu}
  \end{align*}
  where we used Eq.~\eqref{eq:Xi-Xi_mu}, the relation $\pi_{\mu} \circ {\rm shift}_{\mu}^{-1} \circ \hl^{\bar{\mathcal{M}}}_{\bar{q}} = \varphi_{\mu}^{-1}$ from the diagram \eqref{dia:gamma-dtildeW_mu}, and the definition of $\tilde{\Xi}_{\mu}$ from Proposition~\ref{prop:FurtherReduction-K_mu=K}.
  This implies
  \begin{equation*}
    \bar{\gamma}_{\mu}^{*} \Xi(Y, Z) = 
    (d\tilde{W}^{\mu})^{*} \circ \tilde{\Psi}_{1/f_{\mu}}^{*} \tilde{\Xi}_{\mu}(\tilde{Y}, \tilde{Z}).
  \end{equation*}

  As a result, Eq.~\eqref{eq:dgamma2} becomes
  \begin{equation*}
    d\gamma(Y^{\rm h},Z^{\rm h}) = 
    -\frac{1}{f_{\mu}^{2}} (d\tilde{W}^{\mu})^{*} \brackets{ df_{\mu} \wedge \tilde{\Theta} - f_{\mu}^{2} \parentheses{ B^{K}_{\mu} + \tilde{\Psi}_{1/f_{\mu}}^{*} \tilde{\Xi}_{\mu} } }(\tilde{Y}, \tilde{Z}),
  \end{equation*}
  which vanishes because the sufficient condition, Eq.~\eqref{eq:f_mu}, is assumed to be satisfied.
\end{proof}

\subsection{Example of Further Reduction, Hamiltonization, and Chaplygin H--J Equation}
\begin{example}[The Snakeboard; see, e.g., \citet{OsLeMuBu1994}, \citet{BlKrMaMu1996} and \citet{KoMa1997c}]
  \label{ex:Snakeboard}
  Consider the motion of the snakeboard shown in Fig.~\ref{fig:Snakeboard}.
  \begin{figure}[htbp]
    \centering
    \includegraphics[width=.65\linewidth]{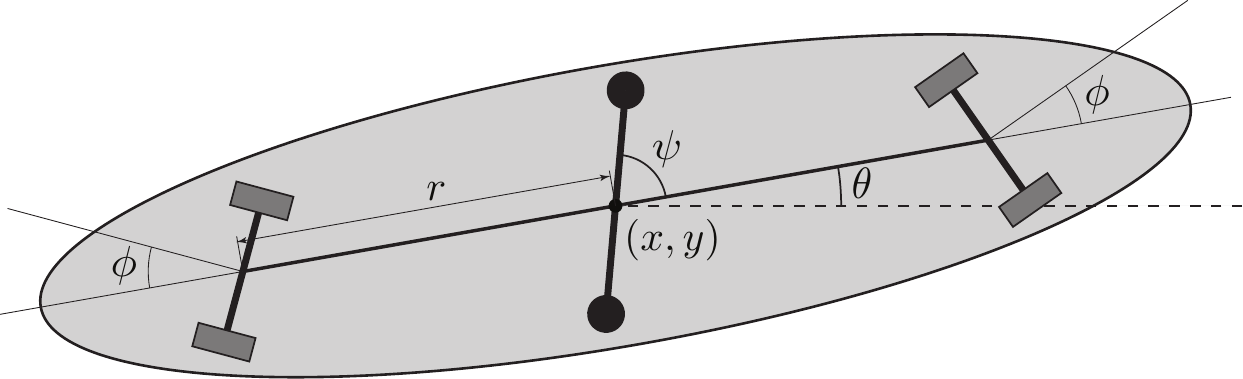}
    \caption{The Snakeboard.}
    \label{fig:Snakeboard}
  \end{figure}
  Let $m$ be the total mass of the board, $J$ the inertia of the board, $J_{0}$ the inertia of the rotor, $J_{1}$ the inertia of each of the wheels, and assume the relation $J + J_{0} + 2 J_{1} = m r^{2}$.
  The configuration space is
  \begin{equation*}
    Q = SE(2) \times \mathbb{S}^{1} \times \mathbb{S}^{1}
    = (SO(2) \ltimes \R^{2}) \times \mathbb{S}^{1} \times \mathbb{S}^{1}
    = \{ (\theta, x, y, \phi, \psi) \}.
  \end{equation*}
  The Lagrangian $L: TQ \to \R$ and the Hamiltonian $H: T^{*}Q \to \R$ are given by
  \begin{equation*}
    L = \frac{1}{2} \brackets{
      m \parentheses{ \dot{x}^{2} + \dot{y}^{2} + r^{2}\dot{\theta}^{2} }
      + 2J_{0}\,\dot{\theta}\,\dot{\psi}
      + 2J_{1}\,\dot{\phi}^{2}
      + J_{0}\,\dot{\psi}^{2}
    }
  \end{equation*}
  and
  \begin{equation*}
    H = \frac{1}{2m}\parentheses{ p_{x}^{2} + p_{y}^{2} } + \frac{1}{2(m r^{2} - J_{0})}(p_{\theta} - p_{\psi})^{2} + \frac{1}{4J_{1}}p_{\phi}^{2} + \frac{1}{2J_{0}}p_{\psi}^{2}.
  \end{equation*}
  The velocity constraints are
  \begin{equation*}
    \dot{x} + r \cot\phi\,\cos\theta\,\dot{\theta} = 0,
    \qquad
    \dot{y} + r \cot\phi\,\sin\theta\,\dot{\theta} = 0,
  \end{equation*}
  or in terms of constraint one-forms,
  \begin{equation*}
    \omega^{1} = dx + r \cot\phi\,\cos\theta\,d\theta,
    \qquad
    \omega^{2} = dy + r \cot\phi\,\sin\theta\,d\theta.
  \end{equation*}
  So the constraint distribution $\mathcal{D} \subset TQ$ and the constrained momentum space $\mathcal{M} \subset T^{*}Q$ are given by
  \begin{equation*}
    \mathcal{D} = \setdef{(\dot{\theta}, \dot{x}, \dot{y}, \dot{\phi}, \dot{\psi}) \in TQ }{ \omega^{s}(\dot{\theta}, \dot{x}, \dot{y}, \dot{\phi}, \dot{\psi}) = 0,\, s = 1,2 },
  \end{equation*}
  and
  \begin{equation*}
    \mathcal{M}
    = \setdef{ (p_{\theta}, p_{x}, p_{y}, p_{\phi}, p_{\psi}) \in T^{*}Q }{
    p_{x} = -\kappa \cot\phi\,\cos\theta\,(p_{\theta} - p_{\psi}),
    \
    p_{y} = -\kappa \cot\phi\,\sin\theta\,(p_{\theta} - p_{\psi})
    },
  \end{equation*}
  where $\kappa \defeq m r/(m r^{2}-J_{0})$.
  
  Let $G = \R^{2}$ and consider the action of $G$ on $Q$ defined by
  \begin{equation*}
    G \times Q \to Q;
    \quad
    \parentheses{(a, b), (\theta, x, y, \phi, \psi)} \mapsto (\theta, x + a, y + b, \phi, \psi).
  \end{equation*}
  Then, the system is a Chaplygin system in the sense of Definition~\ref{def:ChaplyginSystems}.
  The Lie algebra $\mathfrak{g}$ is identified with $\R^{2}$ in this case.
  Let us use again $(\xi, \eta)$ as the coordinates for $\mathfrak{g}$.
  Then, we may write the connection $\mathcal{A}: TQ \to \mathfrak{g}$ as
  \begin{equation}
    \label{eq:mathcalA-Snakeboard}
    \mathcal{A} = (dx + r \cot\phi\cos\theta \, d\theta) \otimes \pd{}{\xi} + (dy + r \cot\phi \sin\theta \, d\theta) \otimes \pd{}{\eta},
  \end{equation}
  and hence its curvature as
  \begin{equation}
    \label{eq:mathcalB-Snakeboard}
    \mathcal{B} = \parentheses{ r \cos\theta\csc^{2}\phi \,d\theta \wedge d\phi \otimes \pd{}{\xi} + r \sin\theta\csc^{2}\phi\,d\theta \wedge d\phi \otimes \pd{}{\eta} }.
  \end{equation}
  Furthermore, the momentum map ${\bf J}: T^{*}Q \to \mathfrak{g}^{*}$ is given by
  \begin{equation}
    \label{eq:J-Snakeboard}
    {\bf J}(p_{q}) = p_{x}\,d\xi + p_{y}\,d\eta.
  \end{equation}

  The quotient space is $\bar{Q} \defeq Q/G = \{(\theta,\phi,\psi)\}$, and the reduced Hamiltonian $\bar{H}: T^{*}\bar{Q} \to \R$ is
  \begin{equation*}
    \bar{H} = 
    \frac{1}{2} \parentheses{
      \frac{\sin^{2}\phi}{m r^{2} - J_{0}\sin^{2}\phi}\,(p_{\theta} - p_{\psi})^{2}
      + \frac{p_{\phi}^{2}}{2 J_{1}}
      + \frac{p_{\psi}^{2}}{J_{0}}
    },
  \end{equation*}
  and the horizontal lift $\hl^{\mathcal{M}}: T^{*}\bar{Q} \to \mathcal{M}$ is given by
  \begin{multline}
    \label{eq:hl^M-Snakeboard}
   \hl^{\mathcal{M}}(p_{\theta},p_{\phi},p_{\psi}) =
   \biggl(
     p_{\psi} + \frac{(m r^{2} - J_{0}) \sin^{2}\phi}{m r^{2} - J_{0}\sin^{2}\phi}(p_{\theta} - p_{\psi}),
     -\frac{m r \cos\phi \sin\phi \cos\theta}{m r^{2} - J_{0}\sin^{2}\phi}(p_{\theta} - p_{\psi}),
     \\
     {-}\frac{m r \cos\phi \sin\phi \sin\theta}{m r^{2} - J_{0}\sin^{2}\phi}(p_{\theta} - p_{\psi}),
     p_{\phi}, p_{\psi}
   \biggr).
 \end{multline}
  Then, we find from Eq.~\eqref{eq:Xi-def} along with Eqs.~\eqref{eq:mathcalA-Snakeboard}, \eqref{eq:mathcalB-Snakeboard}, \eqref{eq:J-Snakeboard}, and \eqref{eq:hl^M-Snakeboard} that
  \begin{equation*}
    \Xi = -\frac{ m r^{2} (p_{\theta} - p_{\psi}) \cot\phi }{ m r^{2} - J_{0} \sin^{2}\phi }\,d\theta \wedge d\phi.
  \end{equation*}
  However, there exists no function $f$ that satisfies the sufficient condition, Eq.~\eqref{eq:f}, for Chaplygin Hamiltonization.
  In fact, one can show (see \cite{FeMeBl2009}) that there does not exist an $f$ which satisfies the necessary and sufficient condition for Hamiltonization from Proposition~\ref{prop:NSCondition-Hamiltonization}.
  Hence the system is not Hamiltonizable at this level of reduction.
  Therefore, we would like to further reduce the system:
  Let $K = \mathbb{S}^{1}$ and consider the action of $K$ on $\bar{Q}$ defined by
  \begin{equation*}
    K \times \bar{Q} \rightarrow \bar{Q};
    \quad
    \parentheses{ c, (\theta, \phi, \psi) } \mapsto (\theta, \phi, \psi + c);
  \end{equation*}
  and so $\Phi^{K}_{c}(\theta, \phi, \psi) = (\theta,\phi, \psi + c)$ for any $c \in K$.
  This gives rise to the cotangent lift
  \begin{equation*}
    K \times T^{*}\bar{Q} \rightarrow T^{*}\bar{Q};
    \quad
    \parentheses{ c, (\theta, \phi, \psi, p_{\theta}, p_{\phi}, p_{\psi}) } \mapsto (\theta, \phi, \psi + c, p_{\theta}, p_{\phi}, p_{\psi}),
  \end{equation*}
  that is,
  \begin{equation*}
    T^{*}\Phi^{K}_{-c}(\theta, \phi, \psi, p_{\theta}, p_{\phi}, p_{\psi}) = (\theta, \phi, \psi + c, p_{\theta}, p_{\phi}, p_{\psi}).
  \end{equation*}
  It is easy to see that the Hamiltonian $\bar{H}$ is $K$-invariant.
  Also, for any $\zeta \in \mathfrak{k}$, we have the infinitesimal generator $\zeta_{T^{*}\bar{Q}} = \zeta\,\tpd{}{\psi}$ and so easily see that $i_{\zeta_{T^{*}\bar{Q}}} \Xi = 0$.
  Hence Conditions~\ref{asmptn:K-invariance} and \ref{asmptn:i_eta-Xi} are satisfied.
  Therefore, by Proposition~\ref{prop:J_K-conservation}, the corresponding momentum map
  \begin{equation*}
    {\bf J}_{K}(p_{\bar{q}}) = p_{\psi}\,d\zeta
  \end{equation*}
  is conserved.
  It is straightforward to check that Condition~\ref{asmptn:K-free_and_proper} is satisfied for any $\mu = \mu_{\psi}\,d\zeta \in \mathfrak{k}^{*}$.
  Then, Proposition~\ref{prop:FurtherReduction} gives the reduced dynamics on ${\bf J}_{K}^{-1}(\mu)/K_{\mu}$, and Eqs.~\eqref{eq:barH-barH_mu} and \eqref{eq:Xi-Xi_mu} give
  \begin{equation*}
    \bar{H}_{\mu} = 
    \frac{1}{2} \parentheses{
      \frac{\sin^{2}\phi}{m r^{2} - J_{0}\sin^{2}\phi}\,(p_{\theta} - \mu_{\psi})^{2}
      + \frac{p_{\phi}^{2}}{2 J_{1}}
      + \frac{\mu_{\psi}^{2}}{J_{0}}
    }
  \end{equation*}
  and
  \begin{equation*}
    \Xi_{\mu} = -\frac{ m r^{2} (p_{\theta} - \mu_{\psi}) \cot\phi }{ m r^{2} - J_{0} \sin^{2}\phi }\,d\theta \wedge d\phi.
  \end{equation*}
  Furthermore, Eq.~\eqref{eq:mathcalA_K} gives the mechanical connection
  \begin{equation}
    \label{eq:mathcalA_K-Snakeboard}
    \mathcal{A}_{K} = (d\theta + d\psi) \otimes \pd{}{\zeta},
  \end{equation}
  and hence Eq.~\eqref{eq:alpha_mu} gives
  \begin{equation}
    \label{eq:alpha_mu-Snakeboard}
    \alpha_{\mu} = \mu (d\theta + d\psi),
  \end{equation}
  and so $\beta_{\mu} = 0$ and $B^{K}_{\mu} = 0$.
  It is also straightforward to check that Conditions~\ref{asmptn:K_mu=K} and \ref{asmptn:alpha_mu} are satisfied.
  Therefore, we may apply Proposition~\ref{prop:FurtherReduction-K_mu=K} to this case.
  Specifically, we have $\tilde{Q} \defeq \bar{Q}/K = \{ (\theta, \phi) \}$, and Eq.~\eqref{eq:varphi_mu-Snakeboard} (from Example~\ref{ex:varphi_mu-Snakeboard} in Appendix~\ref{sec:varphi_mu}) gives 
  \begin{equation*}
    \varphi_{\mu}^{-1}: T^{*}\tilde{Q} \to {\bf J}_{K}^{-1}(\mu)/K;
    \quad
    (\theta, \phi, p_{\theta}, p_{\phi}) \mapsto (\theta, \phi, p_{\theta} + \mu_{\psi}, p_{\phi}),
  \end{equation*}
  and hence we have
  \begin{equation*}
    \tilde{H}_{\mu}(\theta, \phi, p_{\theta}, p_{\phi})
    \defeq \bar{H}_{\mu} \circ \varphi_{\mu}^{-1}(\theta, \phi, p_{\theta}, p_{\phi})
    = \frac{1}{2} \parentheses{
      \frac{\sin^{2}\phi}{m r^{2} - J_{0}\sin^{2}\phi}\,p_{\theta}^{2}
      + \frac{p_{\phi}^{2}}{2 J_{1}}
      + \frac{\mu_{\psi}^{2}}{J_{0}}
    }
  \end{equation*}
  and
  \begin{equation*}
    \tilde{\Xi}_{\mu}
    \defeq (\varphi_{\mu}^{-1})^{*} \Xi_{\mu}
    = -\frac{ m r^{2} p_{\theta} \cot\phi }{ m r^{2} - J_{0} \sin^{2}\phi }\,d\theta \wedge d\phi.
  \end{equation*}
  Therefore, the sufficient condition, Eq.~\eqref{eq:f_mu}, for Chaplygin Hamiltonization becomes
  \begin{equation*}
    p_{\phi} \pd{f_{\mu}}{\theta} - p_{\theta} \pd{f_{\mu}}{\phi}
    =
    -p_{\theta}\, \frac{ m r^{2} \cot\phi }{ m r^{2} - J_{0} \sin^{2}\phi }\,f_{\mu},
  \end{equation*}
  which gives
  \begin{equation*}
    \pd{f_{\mu}}{\theta} = 0,
    \qquad
    \pd{f_{\mu}}{\phi} = \frac{ m r^{2} \cot\phi }{ m r^{2} - J_{0} \sin^{2}\phi }\,f_{\mu}.
  \end{equation*}
  A straightforward integration yields\footnote{For $m r^{2} = J_{0} = 1$, this verifies the result of \citep[Section~4.4]{FeMeBl2009}.}
  \begin{equation*}
    f_{\mu} = \frac{\sin\phi}{\sqrt{ m r^{2} - J_{0} \sin^{2}\phi }},
  \end{equation*}
  where we assume that $|\sin\phi| < \sqrt{m/J_{0}}\,r$.
  Then, Eq.~\eqref{eq:ChaplyginHamiltonian2} gives the following Chaplygin Hamiltonian:
  \begin{align*}
    \tilde{H}_{\rm C}^{\mu}(\theta, \phi, p_{\theta}, p_{\phi})
    &= \tilde{H}_{\mu}\parentheses{
      \theta, \phi, \frac{\sqrt{m r^{2} - J_{0} \sin^{2}\phi}}{\sin\phi}\,p_{\theta}, \frac{\sqrt{m r^{2} - J_{0} \sin^{2}\phi}}{\sin\phi}\,p_{\phi}
    }
    \\
    &= \frac{1}{2} \parentheses{
      p_{\theta}^{2}
      + \frac{m r^{2} - J_{0} \sin^{2}\phi}{2 J_{1} \sin^{2}\phi}\,p_{\phi}^{2}
      + \frac{\mu_{\psi}^{2}}{J_{0}}
    }.
  \end{align*}
  Hence the Chaplygin Hamilton--Jacobi equation~\eqref{eq:ChaplyginHJ2} becomes
  \begin{equation}
    \label{eq:ChaplyginHJ-Snakeboard}
    \frac{1}{2} \brackets{
      \parentheses{ \pd{\tilde{W}^{\mu}}{\theta} }^{2}
      + \frac{m r^{2} - J_{0} \sin^{2}\phi}{2 J_{1} \sin^{2}\phi}\,\parentheses{ \pd{\tilde{W}^{\mu}}{\phi} }^{2}
      + \frac{\mu_{\psi}^{2}}{J_{0}}
    }
    = E.
  \end{equation}
  Assume that $\tilde{W}^{\mu}: \tilde{Q} \to \R$ takes the following form:
  \begin{equation*}
    \tilde{W}^{\mu}(\theta, \phi) = \tilde{W}^{\mu}_{\theta}(\theta) + \tilde{W}^{\mu}_{\phi}(\phi).
  \end{equation*}
  Then, Eq.~\eqref{eq:ChaplyginHJ-Snakeboard} becomes
  \begin{equation*}
    \frac{1}{2} \brackets{
      \parentheses{ \od{\tilde{W}^{\mu}_{\theta}}{\theta} }^{2}
      + \frac{m r^{2} - J_{0} \sin^{2}\phi}{2 J_{1} \sin^{2}\phi}\,\parentheses{ \od{\tilde{W}^{\mu}_{\phi}}{\phi} }^{2}
      + \frac{\mu_{\psi}^{2}}{J_{0}}
    }
    = E.
  \end{equation*}
  The first term in the brackets depends only on $\theta$ whereas the second only on $\phi$, and the third one is constant.
  Thus we have
  \begin{equation*}
    \od{\tilde{W}^{\mu}_{\theta}}{\theta} = \gamma_{\theta}^{0},
    \qquad
    \od{\tilde{W}^{\mu}_{\phi}}{\phi} = \frac{ \sin\phi }{ \sqrt{m r^{2} - J_{0} \sin^{2}\phi} }\,\gamma_{\phi}^{0},
  \end{equation*}
  with some set of constants $\gamma_{\theta}^{0}$ and $\gamma_{\phi}^{0}$ that satisfy
  \begin{equation*}
    \frac{1}{2} \parentheses{
      \parentheses{ \gamma_{\theta}^{0} }^{2}
      + \frac{(\gamma_{\phi}^{0})^{2}}{2 J_{1}}
      + \frac{\mu_{\psi}^{2}}{J_{0}}
    }
    = E,
  \end{equation*}
  which is solved for $\gamma_{\theta}^{0}$ (assumed to be positive) to give
  \begin{equation*}
    \gamma_{\theta}^{0} = \sqrt{ 2 \parentheses{ E - \frac{(\gamma_{\phi}^{0})^{2}}{4J1} - \frac{\mu_{\psi}^{2}}{2 J_{0}}} }.
  \end{equation*}
  Therefore, Eq.~\eqref{eq:gamma-bargamma_mu} with Eq.~\eqref{eq:bargamma_mu-dtildeW_mu} gives
  \begin{align*}
    \gamma(\theta, x, y, \phi, \psi) 
    &= \parentheses{ \mu_{\psi} + \frac{(m r^{2} - J_{0})\,C \sin\phi}{ g(\phi) } } d\theta
    \\
    &\quad -\frac{m r C \cot\phi\,\sin\phi}{g(\phi)}\,(\cos\theta\,dx + \sin\theta\,dy)
    + \gamma_{\phi}^{0}\,d\phi + \mu_{\psi}\,d\psi,
  \end{align*}
  where we defined
  \begin{equation*}
    C \defeq \sqrt{ E - \frac{(\gamma_{\phi}^{0})^{2}}{4J_{1}} - \frac{\mu_{\psi}^{2}}{2J_{0}} },
    \qquad
    g(\phi) \defeq \sqrt{(m r^{2} - J_{0}\sin^{2}\phi )/2}.
  \end{equation*}
  This is the solution of the nonholonomic Hamilton--Jacobi equation~\eqref{eq:NHHJ} obtained in \citet[][Example~4.3]{OhBl2009}.
\end{example}


\section{Conclusion and Future Work}
We established a link between two different approaches towards nonholonomic Hamilton--Jacobi theory; the direct one in \cite{IgLeMa2008, OhBl2009} and the indirect one via Hamiltonization.
We formulated the procedure of Hamiltonization in an intrinsic manner; this helped us understand the relationship between the two approaches and also lead us to the formulas relating the solutions of the two different types of Hamilton--Jacobi equations resulting from the direct and indirect approaches.
The formulas provide us with the following new method to exactly integrate equations of motion of nonholonomic systems:
\begin{enumerate}[\sf 1.]
\item Reduce and Hamiltonize the nonholonomic system.
\item Solve the Chaplygin Hamilton--Jacobi equation for the Hamiltonized reduced system.
\item Use the formula in Theorem~\ref{thm:ChaplyginHJ-NHHJ} or \ref{thm:ChaplyginHJ2-NHHJ} to obtain the solution of the nonholonomic Hamilton--Jacobi equation for the {\em full} dynamics.
\item Integrate the full dynamics using the solution as shown in \citet{OhBl2009}.
\end{enumerate}
A notable feature of this method is that it links the solution of the Hamilton--Jacobi equation for the {\em reduced} system with integration of the {\em full} dynamics.
We illustrated this method with a few examples and obtained the solutions identical to those in \citet{OhBl2009}.

The following questions are interesting to consider for future work:
\begin{itemize}
\item {\em Hamiltonization and Hamilton--Jacobi theory for a more general class of nonholonomic systems with symmetries}:
  This paper only dealt with Chaplygin systems, a special case of the more general class of nonholonomic systems with symmetries treated in \citet{BlKrMaMu1996}.
  We are interested in extending our results to the general case, possibly relating them to the results on existence of an invariant measure in \cite{ZeBl2003}.
  \smallskip
\item {\em Application to nonholonomic systems on Lie groups}:
  Nonholonomic systems on Lie groups, such as the Suslov problem (see, e.g., \citet{Ko1988, ZeBl2000}), often involve an interesting question on integrability: Whether or not the full dynamics is integrable when the reduced dynamics is (see \citet{FeMaPr2009}).
  Relating this question with the Hamilton--Jacobi equations for the full and reduced dynamics is an interesting question to consider.
\end{itemize}

\section*{Acknowledgments}
A.M.~Bloch was supported by the NSF grant DMS-0907949.
O.E.~Fernandez was supported by the Institute for Mathematics and its Applications, through its Postdoctoral Fellowship program, and by the Michigan AGEP Alliance Fellowship.
D.V.~Zenkov was supported by the NSF grants DMS-0604108 and DMS-0908995.
We would like to thank Melvin Leok, Joris Vankerschaver, Hiroaki Yoshimura, and the referee for their many useful comments.


\appendix

\section{Some Lemmas on the Horizontal Lift $\hl^{\mathcal{M}}$}
\begin{lemma}
  \label{lem:G-hl^M}
  The horizontal lift $\hl^{\mathcal{M}}$ is invariant under the action of the cotangent lift of $\Phi$.
  Specifically, for any $h \in G$, we have 
  \begin{equation}
    \label{eq:G-hl^M}
    \hl^{\mathcal{M}}_{h q} = T^{*}_{q}\Phi_{h^{-1}} \circ \hl^{\mathcal{M}}_{q},
  \end{equation}
  where $h q = \Phi_{h}(q)$; or equivalently, for any $\alpha_{\bar{q}} \in T^{*}_{\bar{q}}\bar{Q}$,
  \begin{equation*}
    \alpha^{\rm h}_{h q} = T^{*}_{q}\Phi_{h^{-1}} (\alpha^{\rm h}_{q}).
  \end{equation*}
\end{lemma}

\begin{proof}
  From the definition of $\hl^{\mathcal{M}}_{q}$ and the $G$-invariance of $\hl^{\mathcal{D}}$, we have
  \begin{align*}
    \hl^{\mathcal{M}}_{h q}
    &= \FL_{h q} \circ \hl^{\mathcal{D}}_{h q} \circ (\F\bar{L})^{-1}_{\bar{q}}
    \\
    &= \FL_{h q} \circ T_{q}\Phi_{h} \circ \hl^{\mathcal{D}}_{q} \circ (\F\bar{L})^{-1}_{\bar{q}}.
  \end{align*}
  Now, using the $G$-invariance of the Lagrangian $L$, we have, for any $v_{q} \in T_{q}Q$ and $w_{h q} \in T_{h q}Q$, 
  \begin{align*}
    \ip{ \FL_{h q} \circ T_{q}\Phi_{h} (v_{q}) }{ w_{h q} }
    &=  \left.\od{}{\varepsilon} L( T_{q}\Phi_{h}(v_{q}) + \varepsilon\,w_{h q}) \right|_{\varepsilon=0}
    \\
    &=  \left.\od{}{\varepsilon} L \circ T_{q}\Phi_{h} ( v_{q} + \varepsilon\,T_{h q}\Phi_{h^{-1}}(w_{h q}) ) \right|_{\varepsilon=0}
    \\
    &=  \left.\od{}{\varepsilon} L ( v_{q} + \varepsilon\,T_{h q}\Phi_{h^{-1}}(w_{h q}) ) \right|_{\varepsilon=0}
    \\
    &= \ip{ \FL_{q}(v_{q}) }{ T_{h q}\Phi_{h^{-1}}(w_{h q}) }
    \\
    &= \ip{ T^{*}_{q}\Phi_{h^{-1}}(\FL_{q}(v_{q})) }{ w_{h q} },
  \end{align*}
  and thus $\FL_{h q} \circ T\Phi_{h} = T^{*}\Phi_{h^{-1}} \circ \FL_{q}$.
  Hence we obtain
  \begin{align*}
    \hl^{\mathcal{M}}_{h q}
    &= T^{*}_{q}\Phi_{h^{-1}} \circ \FL_{q} \circ \hl^{\mathcal{D}}_{q} \circ (\F\bar{L})^{-1}_{\bar{q}}
    \\
    &= T^{*}_{q}\Phi_{h^{-1}} \circ \hl^{\mathcal{M}}_{q}.
    \qedhere
  \end{align*}
\end{proof}

\begin{lemma}
  \label{lem:hl-pairing}
  Let $q$ be an arbitrary point in $Q$ and $\bar{q} = \pi(q) \in \bar{Q}$.
  For any $\alpha_{\bar{q}} \in T^{*}_{\bar{q}}\bar{Q}$ and $v_{\bar{q}} \in T_{\bar{q}}\bar{Q}$, the following identity holds:
  \begin{equation*}
    \ip{ \hl^{\mathcal{M}}_{q}(\alpha_{\bar{q}}) }{ \hl^{\mathcal{D}}_{q}(v_{\bar{q}}) } = \ip{ \alpha_{\bar{q}} }{ v_{\bar{q}} }.
  \end{equation*}
\end{lemma}

\begin{proof}
  Follows from the definitions of $\bar{g}$ and $\hl^{\mathcal{M}}_{q}$ (see Eqs.~\eqref{eq:barg} and \eqref{eq:hl^M}, respectively):
  \begin{align*}
    \ip{ \hl^{\mathcal{M}}_{q}(\alpha_{\bar{q}}) }{ \hl^{\mathcal{D}}_{q}(v_{\bar{q}}) }
    &= \ip{ g_{q}^{\flat} \circ \hl^{\mathcal{D}}_{q} \circ (\bar{g}^{\flat})^{-1}_{\bar{q}}(\alpha_{\bar{q}}) }{ \hl^{\mathcal{D}}_{q}(v_{\bar{q}}) }
    \\
    &= g_{q}\parentheses{ \hl^{\mathcal{D}}_{q} \circ (\bar{g}^{\flat})^{-1}_{\bar{q}}(\alpha_{\bar{q}}), \hl^{\mathcal{D}}_{q}(v_{\bar{q}}) }
    \\
    &= \bar{g}_{\bar{q}}\parentheses{ (\bar{g}^{\flat})^{-1}_{\bar{q}}(\alpha_{\bar{q}}), v_{\bar{q}} }
    \\
    &= \ip{ \bar{g}^{\flat}_{\bar{q}} \circ (\bar{g}^{\flat})^{-1}_{\bar{q}}(\alpha_{\bar{q}})}{ v_{\bar{q}} }
   \\
    &= \ip{ \alpha_{\bar{q}} }{ v_{\bar{q}} }. \qedhere
  \end{align*}
\end{proof}

\section{Construction of $\varphi_{\mu}: {\bf J}_{K}^{-1}(\mu)/K \to T^{*}\tilde{Q}$}
\label{sec:varphi_mu}
We briefly summarize the construction of the map $\varphi_{\mu}: {\bf J}_{K}^{-1}(\mu)/K \to T^{*}\tilde{Q}$ that appears in Proposition~\ref{prop:FurtherReduction-K_mu=K} following \citet[Section~2.2]{MaMiOrPeRa2007}.
First define $\bar{\varphi}_{0}: {\bf J}_{K}^{-1}(0) \to T^{*}\tilde{Q}$ by
\begin{equation}
  \label{eq:barvarphi_0}
  \ip{ \bar{\varphi}_{0}(p_{\bar{q}}) }{ T_{\bar{q}}\bar{\pi}(v_{\bar{q}}) } = \ip{ p_{\bar{q}} }{ v_{\bar{q}} }
\end{equation}
for any $p_{\bar{q}} \in T_{\bar{q}}\bar{Q}$ and $v_{\bar{q}} \in T_{\bar{q}}\bar{Q}$.
Let $\pi_{0}: {\bf J}_{K}^{-1}(0) \to {\bf J}_{K}^{-1}(0)/K$ be the projection to the quotient.
Then, $\varphi_{0}: {\bf J}_{K}^{-1}(0)/K \to T^{*}\tilde{Q}$ is uniquely characterized by the relation
\begin{equation}
  \label{eq:varphi_0}
  \varphi_{0} \circ \pi_{0} = \bar{\varphi}_{0}.
\end{equation}
It can be shown that $\varphi_{0}$ is in fact a diffeomorphism (see \citet[Proof of Theorem~2.2.2 on pp.~62--63]{MaMiOrPeRa2007}).
We also introduce the shift map
\begin{equation*}
  {\rm Shift}_{\mu}: T^{*}\bar{Q} \to T^{*}\bar{Q}
\end{equation*}
defined by
\begin{equation}
  \label{eq:Shift_mu}
  {\rm Shift}_{\mu}(p_{\bar{q}}) \defeq p_{\bar{q}} - \alpha_{\mu}(\bar{q}),
\end{equation}
where $\alpha_{\mu}$ is the one-form on $\bar{Q}$ defined in Eq.~\eqref{eq:alpha_mu}.
This gives rise to the $K$-equivariant diffeomorphism
\begin{equation*}
  {\rm shift}_{\mu}: {\bf J}_{K}^{-1}(\mu) \to {\bf J}_{K}^{-1}(0),
\end{equation*}
and the commutative diagram below, where $i_{\mu}$ and $i_{0}$ are both inclusions.
\begin{equation*}
  \vcenter{
    \xymatrix@!0@R=0.75in@C=1in{
      T^{*}\bar{Q} \ar[r]^{{\rm Shift}_{\mu}} & T^{*}\bar{Q}
      \\
      {\bf J}_{K}^{-1}(\mu) \ar[u]^{i_{\mu}} \ar[r]_{{\rm shift}_{\mu}} & {\bf J}_{K}^{-1}(0) \ar[u]_{i_{0}}
    }
  }
\end{equation*}
Since the map ${\rm shift}_{\mu}$ is $K$-equivariant, it induces the diffeomorphism
\begin{equation*}
  \widetilde{\rm shift}_{\mu}: {\bf J}_{K}^{-1}(\mu)/K \to {\bf J}_{K}^{-1}(0)/K.
\end{equation*}
The map $\varphi_{\mu}: {\bf J}_{K}^{-1}(\mu)/K \to T^{*}\tilde{Q}$ is then defined by
\begin{equation}
  \label{eq:varphi_mu}
  \varphi_{\mu} \defeq \varphi_{0} \circ \widetilde{\rm shift}_{\mu}.
\end{equation}
The diagram below summarizes the construction of $\varphi_{\mu}$.
\begin{equation*}
  \vcenter{
    \xymatrix@!0@R=0.9in@C=1.15in{
      {\bf J}_{K}^{-1}(\mu) \ar[r]^{{\rm shift}_{\mu}} \ar[d]_{\pi_{\mu}} & {\bf J}_{K}^{-1}(0) \ar[d]^{\pi_{0}} \ar[rd]^{\bar{\varphi}_{0}} & 
      \\
      {\bf J}_{K}^{-1}(\mu)/K \ar[r]_{\widetilde{\rm shift}_{\mu}} \ar@/_{2pc}/[rr]^{\varphi_{\mu}} & {\bf J}_{K}^{-1}(0)/K \ar[r]_{\varphi_{0}} & T^{*}\tilde{Q}
    }
  }
\end{equation*}

\begin{example}[The Snakeboard; see Example~\ref{ex:Snakeboard}]
  \label{ex:varphi_mu-Snakeboard}
  Let us first determine $\bar{\varphi}_{0}$ and $\varphi_{0}$.
  Note that we may parametrize ${\bf J}_{K}^{-1}(0)$ as follows:
  \begin{equation*}
    {\bf J}_{K}^{-1}(0)
    = \setdef{ (\theta, \phi, \psi, p_{\theta}, p_{\phi}, p_{\psi}) \in T^{*}\bar{Q} }{ p_{\psi} = 0 }
    = \{ (\theta, \phi, \psi, p_{\theta}, p_{\phi}) \},
  \end{equation*}
  and also that $\bar{\pi}(\theta, \phi, \psi) = (\theta, \phi)$ and hence $T\bar{\pi}(v_{\theta}, v_{\phi}, v_{\psi}) = (v_{\theta} , v_{\phi})$.
  Therefore, Eq.~\eqref{eq:barvarphi_0} gives
  \begin{equation*}
    \bar{\varphi}_{0}(\theta, \phi, \psi, p_{\theta}, p_{\phi}) = (\theta, \phi, p_{\theta}, p_{\phi}).
  \end{equation*}
  Since $\pi_{0}(\theta, \phi, \psi, p_{\theta}, p_{\phi}) = (\theta, \phi, p_{\theta}, p_{\phi})$, Eq.~\eqref{eq:varphi_0} gives
  \begin{equation*}
    \varphi_{0}(\theta, \phi, p_{\theta}, p_{\phi}) = (\theta, \phi, p_{\theta}, p_{\phi}).
  \end{equation*}
  Now, let us determine the map $\widetilde{\rm shift}_{\mu}$: Using the $\alpha_{\mu}$ in Eq.~\eqref{eq:alpha_mu-Snakeboard}, we find, from Eq.~\eqref{eq:Shift_mu},
  \begin{equation*}
    {\rm Shift}_{\mu}(\theta, \phi, \psi, p_{\theta}, p_{\phi}, p_{\psi}) = (\theta, \phi, \psi, p_{\theta} - \mu_{\psi}, p_{\phi}, p_{\psi} - \mu_{\psi}).
  \end{equation*}
  Parameterizing ${\bf J}_{K}^{-1}(\mu)$ as
  \begin{equation*}
    {\bf J}_{K}^{-1}(\mu)
    = \setdef{ (\theta, \phi, \psi, p_{\theta}, p_{\phi}, p_{\psi}) \in T^{*}\bar{Q} }{ p_{\psi} = \mu_{\psi} }
    = \{ (\theta, \phi, \psi, p_{\theta}, p_{\phi}) \},
  \end{equation*}
  we obtain
  \begin{equation*}
    {\rm shift}_{\mu}(\theta, \phi, \psi, p_{\theta}, p_{\phi}) = (\theta, \phi, \psi, p_{\theta} - \mu_{\psi}, p_{\phi}),
  \end{equation*}
  and hence
  \begin{equation*}
    \widetilde{\rm shift}_{\mu}(\theta, \phi, p_{\theta}, p_{\phi}) = (\theta, \phi, p_{\theta} - \mu_{\psi}, p_{\phi}).
  \end{equation*}
  As a result, we obtain, from Eq.~\eqref{eq:varphi_mu},
  \begin{equation}
    \label{eq:varphi_mu-Snakeboard}
    \varphi_{\mu}(\theta, \phi, p_{\theta}, p_{\phi}) = (\theta, \phi, p_{\theta} - \mu_{\psi}, p_{\phi}).
  \end{equation}
\end{example}

\section{On the Horizontal Lift $\hl^{\bar{\mathcal{M}}}$}
\label{sec:hl^barM}

\begin{lemma}
  \label{lem:hl^barM}
  Let $\bar{q}$ be a point in $\bar{Q}$ and $\tilde{q} = \bar{\pi}(\bar{q})$.
  Then, we have $\hl^{\bar{\mathcal{M}}}_{\bar{q}} \circ \bar{\varphi}_{0}(p_{\bar{q}}) = p_{\bar{q}}$ for any $p_{\bar{q}} \in {\bf J}_{K}^{-1}(0)$ and also $\bar{\varphi}_{0} \circ \hl^{\bar{\mathcal{M}}} = \id_{T^{*}\tilde{Q}}$, and the diagram
  \begin{equation*}
    \vcenter{
      \xymatrix@!0@R=0.9in@C=1.15in{
        {\bf J}_{K}^{-1}(0) \ar[d]_{\pi_{0}} & 
        \\
        {\bf J}_{K}^{-1}(0)/K \ar[r]_{\varphi_{0}} & T^{*}\tilde{Q} \ar[ul]_{\hl^{\bar{\mathcal{M}}}}
      }
    }
  \end{equation*}
  commutes with an appropriate choice of the base point $\bar{q}$ of the image of $\hl^{\bar{\mathcal{M}}}$.
\end{lemma}

\begin{proof}
  Let $p_{\tilde{q}} \in T_{\tilde{q}}^{*}\tilde{Q}$ and $v_{\tilde{q}} \in T_{\tilde{q}}\tilde{Q}$ be arbitrary.
  Then, Eq.~\eqref{eq:barvarphi_0} implies that
  \begin{align*}
    \ip{ \bar{\varphi}_{0} \circ \hl^{\bar{\mathcal{M}}}_{\bar{q}}(p_{\tilde{q}}) }{ T_{\bar{q}}\bar{\pi}\parentheses{ \hl^{\bar{\mathcal{D}}}_{\bar{q}}(v_{\tilde{q}}) } }
    &= \ip{ \hl^{\bar{\mathcal{M}}}_{\bar{q}}(p_{\tilde{q}}) }{ \hl^{\bar{\mathcal{D}}}_{\bar{q}}(v_{\tilde{q}}) }
    \\
    &= \ip{ p_{\tilde{q}} }{ v_{\tilde{q}} },
  \end{align*}
  where we used an identity on pairings between the images of $\hl^{\bar{\mathcal{M}}}$ and $\hl^{\bar{\mathcal{D}}}$, which can be shown in the same way as Lemma~\ref{lem:hl-pairing}.
  However, the definition $\hl^{\bar{\mathcal{D}}} \defeq (T\bar{\pi}|_{\bar{\mathcal{D}}})^{-1}$ implies $T\bar{\pi} \circ \hl^{\bar{\mathcal{D}}} = \id_{T\tilde{Q}}$, and thus the above equation reduces to
  \begin{equation*}
    \ip{ \bar{\varphi}_{0} \circ \hl^{\bar{\mathcal{M}}}_{\bar{q}}(p_{\tilde{q}}) }{ v_{\tilde{q}} }
    = \ip{ p_{\tilde{q}} }{ v_{\tilde{q}} }.
  \end{equation*}
  Therefore, we have $\bar{\varphi}_{0} \circ \hl^{\bar{\mathcal{M}}} = \id_{T^{*}\tilde{Q}}$.
  So Eq.~\eqref{eq:varphi_0} gives
  \begin{equation*}
    \varphi_{0} \circ \pi_{0} \circ \hl^{\bar{\mathcal{M}}} = \id_{T^{*}\tilde{Q}},
  \end{equation*}
  which also implies
  \begin{equation*}
    \pi_{0} \circ \hl^{\bar{\mathcal{M}}} \circ \varphi_{0} = \id_{{\bf J}_{K}^{-1}(0)/K},
  \end{equation*}
  since $\varphi_{0}$ is a diffeomorphism.

  To show $\hl^{\bar{\mathcal{M}}}_{\bar{q}} \circ \bar{\varphi}_{0}(p_{\bar{q}}) = p_{\bar{q}}$, take an arbitrary $v_{\bar{q}} \in T_{\bar{q}}\bar{Q}$.
  Then, we may decompose $v_{\bar{q}}$ into the horizontal and vertical parts, i.e.,
  \begin{equation*}
    v_{\bar{q}} = \hl^{\bar{\mathcal{D}}}_{\bar{q}}(\tilde{v}_{\tilde{q}}) + (\mathcal{A}_{K}(v_{\bar{q}}) )_{\bar{Q}}(\bar{q})
  \end{equation*}
  where $\tilde{v}_{\tilde{q}} = T_{\bar{q}}\bar{\pi}(v_{\bar{q}})$.
  Therefore, we obtain
  \begin{align*}
    \ip{ \hl^{\bar{\mathcal{M}}}_{\bar{q}} \circ \bar{\varphi}_{0}(p_{\bar{q}}) }{ v_{\bar{q}} }
    &= \ip{ \hl^{\bar{\mathcal{M}}}_{\bar{q}} \circ \bar{\varphi}_{0}(p_{\bar{q}}) }{ \hl^{\bar{\mathcal{D}}}_{\bar{q}}(\tilde{v}_{\tilde{q}}) }
    + \ip{ \hl^{\bar{\mathcal{M}}}_{\bar{q}} \circ \bar{\varphi}_{0}(p_{\bar{q}}) }{ (\mathcal{A}_{K}(v_{\bar{q}}) )_{\bar{Q}}(\bar{q}) }
    \\
    &= \ip{ \bar{\varphi}_{0}(p_{\bar{q}}) }{ \tilde{v}_{\tilde{q}} }
    + \ip{ {\bf J}_{K} \circ \hl^{\bar{\mathcal{M}}}_{\bar{q}} \circ \bar{\varphi}_{0}(p_{\bar{q}}) }{ \mathcal{A}_{K}(v_{\bar{q}}) }
    \\
    &= \ip{ \bar{\varphi}_{0}(p_{\bar{q}}) }{ T_{\bar{q}}\bar{\pi}(v_{\bar{q}}) }
    \\
    &= \ip{ p_{\bar{q}} }{ v_{\bar{q}} },
  \end{align*}
  where we used the fact that $\hl^{\bar{\mathcal{M}}}$ takes values in $\bar{\mathcal{M}} \defeq {\bf J}_{K}^{-1}(0)$, and also Eq.~\eqref{eq:barvarphi_0}.
  Hence $\hl^{\bar{\mathcal{M}}}_{\bar{q}} \circ \bar{\varphi}_{0}(p_{\bar{q}}) = p_{\bar{q}}$ and also, by Eq.~\eqref{eq:varphi_0}, $\hl^{\bar{\mathcal{M}}}_{\bar{q}} \circ \varphi_{0} \circ \pi_{0}(p_{\bar{q}}) = p_{\bar{q}}$.
\end{proof}


\bibliography{ChaplyginHJ}

\begin{thebibliography}{30}
\providecommand{\natexlab}[1]{#1}
\providecommand{\url}[1]{\texttt{#1}}
\expandafter\ifx\csname urlstyle\endcsname\relax
  \providecommand{\doi}[1]{doi: #1}\else
  \providecommand{\doi}{doi: \begingroup \urlstyle{rm}\Url}\fi

\bibitem[Abraham and Marsden(1978)]{AbMa1978}
R.~Abraham and J.~E. Marsden.
\newblock \emph{Foundations of Mechanics}.
\newblock Addison--Wesley, 2nd edition, 1978.

\bibitem[Bates and Sniatycki(1993)]{BaSn1993}
L.~Bates and J.~Sniatycki.
\newblock Nonholonomic reduction.
\newblock \emph{Reports on Mathematical Physics}, 32\penalty0 (1):\penalty0
  99--115, 1993.

\bibitem[Bloch(2003)]{Bl2003}
A.~M. Bloch.
\newblock \emph{Nonholonomic Mechanics and Control}.
\newblock Springer, 2003.

\bibitem[Bloch et~al.(1996)Bloch, Krishnaprasad, Marsden, and
  Murray]{BlKrMaMu1996}
A.~M. Bloch, P.~S. Krishnaprasad, J.~E. Marsden, and R.~M. Murray.
\newblock Nonholonomic mechanical systems with symmetry.
\newblock \emph{Archive for Rational Mechanics and Analysis}, 136:\penalty0
  21--99, 1996.

\bibitem[Bloch et~al.(2009)Bloch, Fernandez, and Mestdag]{BlFeMe2009}
A.~M. Bloch, O.~E. Fernandez, and T.~Mestdag.
\newblock Hamiltonization of nonholonomic systems and the inverse problem of
  the calculus of variations.
\newblock \emph{Reports on Mathematical Physics}, 63\penalty0 (2):\penalty0
  225--249, 2009.

\bibitem[Cantrijn et~al.(2002)Cantrijn, Cort\'es, de~Le\'on, and Mart\'in~de
  Diego]{CaCoLeMa2002}
F.~Cantrijn, J.~Cort\'es, M.~de~Le\'on, and D.~Mart\'in~de Diego.
\newblock On the geometry of generalized {C}haplygin systems.
\newblock \emph{Mathematical Proceedings of the Cambridge Philosophical
  Society}, 132\penalty0 (02):\penalty0 323--351, 2002.

\bibitem[Cari\~nena et~al.(2010)Cari\~nena, Gracia, Marmo, Mart\'inez, Mun\~oz
  Lecanda, and Rom\'an-Roy]{CaGrMaMaMuRo2010}
J.~F. Cari\~nena, X.~Gracia, G.~Marmo, E.~Mart\'inez, M.~C. Mun\~oz Lecanda,
  and N.~Rom\'an-Roy.
\newblock Geometric {H}amilton--{J}acobi theory for nonholonomic dynamical
  systems.
\newblock \emph{International Journal of Geometric Methods in Modern Physics},
  7\penalty0 (3):\penalty0 431--454, 2010.

\bibitem[Chaplygin(2008)]{Ch2008}
S.~Chaplygin.
\newblock On the theory of motion of nonholonomic systems. {T}he
  reducing-multiplier theorem.
\newblock \emph{Regular and Chaotic Dynamics}, 13\penalty0 (4):\penalty0
  369--376, 2008.
\newblock (English translation of Matematicheski\u{i} sbornik, 1911, vol. 28,
  issue 1).

\bibitem[de~Le\'on et~al.(2010)de~Le\'on, Marrero, and Mart\'in~de
  Diego]{LeMaMa2010}
M.~de~Le\'on, J.~C. Marrero, and D.~Mart\'in~de Diego.
\newblock Linear almost {P}oisson structures and {H}amilton--{J}acobi equation.
  {A}pplications to nonholonomic mechanics.
\newblock \emph{Journal of Geometric Mechanics}, 2\penalty0 (2):\penalty0
  159--198, 2010.

\bibitem[Ehlers et~al.(2004)Ehlers, Koiller, Montgomery, and
  Rios]{EhKoMoRi2004}
K.~M. Ehlers, J.~Koiller, R.~Montgomery, and P.~M. Rios.
\newblock Nonholonomic systems via moving frames: {C}artan equivalence and
  {C}haplygin {H}amiltonization.
\newblock In \emph{The Breadth of Symplectic and {P}oisson Geometry}, pages
  75--120. Birkh\"auser, 2004.

\bibitem[Fedorov and Jovanovi\'c(2004)]{FeJo2004}
Y.~N. Fedorov and B.~Jovanovi\'c.
\newblock Nonholonomic {LR} systems as generalized {C}haplygin systems with an
  invariant measure and flows on homogeneous spaces.
\newblock \emph{Journal of Nonlinear Science}, 14\penalty0 (4):\penalty0
  341--381, 2004.

\bibitem[Fedorov and Jovanovi\'c(2009)]{FeJo2009}
Y.~N. Fedorov and B.~Jovanovi\'c.
\newblock Hamiltonization of the generalized {V}eselova {LR} system.
\newblock \emph{Regular and Chaotic Dynamics}, 14\penalty0 (4):\penalty0
  495--505, 2009.

\bibitem[Fedorov et~al.(2009)Fedorov, Maciejewski, and Przybylska]{FeMaPr2009}
Y.~N. Fedorov, A.~J. Maciejewski, and M.~Przybylska.
\newblock The {P}oisson equations in the nonholonomic {S}uslov problem:
  integrability, meromorphic and hypergeometric solutions.
\newblock \emph{Nonlinearity}, 22\penalty0 (9), 2009.

\bibitem[Fernandez et~al.(2009)Fernandez, Mestdag, and Bloch]{FeMeBl2009}
O.~E. Fernandez, T.~Mestdag, and A.~M. Bloch.
\newblock A generalization of {C}haplygin's reducibility theorem.
\newblock \emph{Regular and Chaotic Dynamics}, 14\penalty0 (6):\penalty0
  635--655, 2009.

\bibitem[Hochgerner and Garc\'ia-Naranjo(2009)]{HoGa2009}
S.~Hochgerner and L.~Garc\'ia-Naranjo.
\newblock {$G$}-{C}haplygin systems with internal symmetries, truncation, and
  an (almost) symplectic view of {C}haplygin's ball.
\newblock \emph{Journal of Geometric Mechanics}, 1\penalty0 (1):\penalty0
  35--53, 2009.

\bibitem[Iglesias-Ponte et~al.(2008)Iglesias-Ponte, de~Le\'on, and Mart\'in~de
  Diego]{IgLeMa2008}
D.~Iglesias-Ponte, M.~de~Le\'on, and D.~Mart\'in~de Diego.
\newblock Towards a {H}amilton--{J}acobi theory for nonholonomic mechanical
  systems.
\newblock \emph{Journal of Physics A: Mathematical and Theoretical},
  41\penalty0 (1), 2008.

\bibitem[Kobayashi and Nomizu(1963)]{KoNo1963}
S.~Kobayashi and K.~Nomizu.
\newblock \emph{Foundations of Differential Geometry}.
\newblock Interscience, New York, 1963.

\bibitem[Koiller(1992)]{Ko1992}
J.~Koiller.
\newblock Reduction of some classical non-holonomic systems with symmetry.
\newblock \emph{Archive for Rational Mechanics and Analysis}, 118\penalty0
  (2):\penalty0 113--148, 1992.

\bibitem[Koiller et~al.(2002)Koiller, Rios, and Ehlers]{KoRiEh2002}
J.~Koiller, P.~M. Rios, and K.~M. Ehlers.
\newblock Moving frames for cotangent bundles.
\newblock \emph{Reports on Mathematical Physics}, 49\penalty0 (2-3):\penalty0
  225--238, 2002.

\bibitem[Koon and Marsden(1997)]{KoMa1997c}
W.~S. Koon and J.~E. Marsden.
\newblock The {H}amiltonian and {L}agrangian approaches to the dynamics of
  nonholonomic systems.
\newblock \emph{Reports on Mathematical Physics}, 40\penalty0 (1):\penalty0
  21--62, 1997.

\bibitem[Kozlov(1988)]{Ko1988}
V.~V. Kozlov.
\newblock Invariant measures of {E}uler--{P}oincar\'e equations on {L}ie
  algebras.
\newblock \emph{Functional Analysis and Its Applications}, 22\penalty0
  (1):\penalty0 58--59, 1988.

\bibitem[Kozlov(2002)]{Ko2002}
V.~V. Kozlov.
\newblock On the integration theory of equations of nonholonomic mechanics.
\newblock \emph{Regular and Chaotic Dynamics}, 7\penalty0 (2):\penalty0
  161--176, 2002.
\newblock (Reprint from Advances in mechanics USSR, V. 8, No. 3, pp. 85--107,
  1985).

\bibitem[Marsden and Weinstein(1974)]{MaWe1974}
J.~E. Marsden and A.~Weinstein.
\newblock Reduction of symplectic manifolds with symmetry.
\newblock \emph{Reports on Mathematical Physics}, 5\penalty0 (1):\penalty0
  121--130, 1974.

\bibitem[Marsden et~al.(2007)Marsden, Misiolek, Ortega, Perlmutter, and
  Ratiu]{MaMiOrPeRa2007}
J.~E. Marsden, G.~Misiolek, J.~P. Ortega, M.~Perlmutter, and T.~S. Ratiu.
\newblock \emph{Hamiltonian Reduction by Stages}.
\newblock Springer, 2007.

\bibitem[Ohsawa and Bloch(2009)]{OhBl2009}
T.~Ohsawa and A.~M. Bloch.
\newblock Nonholonomic {H}amilton--{J}acobi equation and integrability.
\newblock \emph{Journal of Geometric Mechanics}, 1\penalty0 (4):\penalty0
  461--481, 2009.

\bibitem[Ostrowski et~al.(1994)Ostrowski, Lewis, Murray, and
  Burdick]{OsLeMuBu1994}
J.~Ostrowski, A.~Lewis, R.~Murray, and J.~Burdick.
\newblock Nonholonomic mechanics and locomotion: {T}he snakeboard example.
\newblock \emph{Robotics and Automation, Proceedings., 1994 IEEE International
  Conference on}, pages 2391--2397 vol.3, 1994.

\bibitem[Planas-Bielsa(2004)]{Pl2004}
V.~Planas-Bielsa.
\newblock Point reduction in almost symplectic manifolds.
\newblock \emph{Reports on Mathematical Physics}, 54\penalty0 (3):\penalty0
  295--308, 2004.

\bibitem[Stanchenko(1989)]{St1989}
S.~V. Stanchenko.
\newblock Non-holonomic {C}haplygin systems.
\newblock \emph{Journal of Applied Mathematics and Mechanics}, 53\penalty0
  (1):\penalty0 11--17, 1989.

\bibitem[Zenkov and Bloch(2000)]{ZeBl2000}
D.~V. Zenkov and A.~M. Bloch.
\newblock Dynamics of the $n$-dimensional {S}uslov problem.
\newblock \emph{Journal of Geometry and Physics}, 34\penalty0 (2):\penalty0
  121--136, 2000.

\bibitem[Zenkov and Bloch(2003)]{ZeBl2003}
D.~V. Zenkov and A.~M. Bloch.
\newblock Invariant measures of nonholonomic flows with internal degrees of
  freedom.
\newblock \emph{Nonlinearity}, 16\penalty0 (5):\penalty0 1793--1807, 2003.

\end{thebibliography}
\bibliographystyle{plainnat}

\end{document}